\documentclass[11pt,draftcls, onecolumn]{IEEEtran} 
\usepackage[numbers,sort&compress]{natbib}
\usepackage{amsmath}
\usepackage{amssymb}
\usepackage{amsfonts}
\usepackage{mathrsfs}
\usepackage[dvips]{graphicx}
\usepackage{verbatim}
\usepackage{subfig}
\usepackage{balance}
\usepackage{caption}

\makeatletter

\newcommand{\Rmnum}[1]{\expandafter\@slowromancap\romannumeral #1@}
\makeatother

\newtheorem{theorem}{Theorem}

\newtheorem{corollary}{Corollary}

\newtheorem{lemma}{Lemma}

\newtheorem{definition}{Definition}

\begin{document}

\title{A Robust Generalized Chinese Remainder Theorem for Two Integers}

\author{
Xiaoping Li, Xiang-Gen Xia, {\it Fellow, IEEE},  Wenjie Wang, {\it Member, IEEE}, Wei Wang

\thanks{This work was partially supported by the NSFC (Nos. 61172092, 61302069), the Research Fund for the Doctoral Programs of Higher Education of China (No. 20130201110014), and the Air Force Office of Scientific Research (AFOSR) under Grant FA9550-12-1-0055. This work was done when Xiaoping Li was visiting the University of Delaware.

Xiaoping Li and Wenjie Wang are with MOE Key Lab for Intelligent Networks and Network Security, Xi'an Jiaotong University, Xi'an, Shaanxi 710049 P. R. China (e-mail: {\tt lixiaoping@stu.xjtu.edu.cn; wjwang@xjtu.edu.cn}).

Xiang-Gen Xia is with Xidian University, Xi'an, China, and the Department of
Electrical and Computer Engineering, University of Delaware, Newark,
DE 19716, USA. (e-mail: {\tt xxia@ee.udel.edu}).

Wei Wang is with College of Information Engineering, Tarim University, Alar, Xinjiang 843300 P. R. China (e-mail: {\tt wangwei.math@gmail.com}).}}

\maketitle
\begin{abstract}
A generalized Chinese remainder theorem (CRT) for multiple integers from residue sets has been studied recently, where the correspondence between the remainders and the integers in each residue set modulo several moduli is not known. A robust CRT has also been proposed lately for robustly reconstruct a single integer from its erroneous remainders. In this paper, we consider the reconstruction problem of two integers from their residue sets, where the remainders are not only out of order but also may have errors. We prove that two integers can be robustly reconstructed if their remainder errors are less than $M/8$, where $M$ is the greatest common divisor (gcd) of all the moduli. We also propose an efficient reconstruction algorithm. Finally, we present some simulations to verify the efficiency of the proposed algorithm. The study is motivated and has applications in the determination of multiple frequencies from multiple undersampled waveforms.
\end{abstract}


%
\IEEEpeerreviewmaketitle

\section{INTRODUCTION}

\IEEEPARstart{T}he tranditional Chinese remainder theorem (CRT)
is to reconstruct a single nonnegative integer from its remainders modulo several smaller positive integers (called moduli) and it has tremendous applications in various areas \cite{bi:CRT1}-\cite{Licq}. There are various generalizations of CRT, see, for example, \cite{bi:CRT2} for some of them. One of the generalizations, generalized CRT, is to determine multiple integers from their residue sets where each residue set is the set of remainders of the multiple integers modulo a modulus and the correspondence between the remainders and the multiple integers is not known,
i.e., each residue set is not ordered. This problem was first studied in \cite{Arazi}. It was later studied independently in \cite{zhou1997}-\cite{wangwei2014},
motivated from multiple frequency determination in multiple undersampled
waveforms. It exists in many engineering applications, such as phase unwrapping in signal processing \cite{xia2001}-\cite{Akhlaq}, multiwavelength optical interferometry \cite{Falaggis_1}, \cite{Falaggis_2}, radar signal processing \cite{wanggenyuan}-\cite{Qi}, mechanical engineering \cite{Beauseroy}, and wireless sensor networks \cite{Chessa}, \cite{Liwenchao_13}.

Usually the moduli in CRT or the generalized CRT mentioned above are required to be pairwise co-prime, which is not robust in the sense that a small error in its remainders may cause a large reconstruction error. Robust reconstruction methods, i.e., robust CRT, for a single integer from its erroneous remainders have been studied and obtained in \cite{huangwan}-\cite{wangcheng_10}. The basic idea for these robust CRT is to include a common factor among all the moduli and then as long as the remainder errors are less than the quarter of the greatest common factor (gcd) of all the moduli, a reconstruction error of the integer will be less than the maximum remainder error. Several robust reconstruction methods have been proposed, for example, searching based robust CRT \cite{gangli3}-\cite{bi:Xiaowei Li}, closed-form robust CRT \cite{wjwang}, \cite{Xuguangwu}, multi-stage robust CRT \cite{Binyang}, \cite{lixiao_2014}, where in \cite{lixiao_2014} the upper bound of the quarter of the gcd has been improved when the remaining integers factorized by the gcd of all the moduli are not necessarily co-prime. All these studies are only for the traditional CRT for single integers. There is no attempt in the literature to robustly reconstruct multiple integers from their erroneous residue sets, i.e., robust generalized CRT, although \cite{xia_Liu} studies the case when most of the residue sets are error free but only a few remainder sets include erroneous remainders and is not in the sense of the robustness in the literature. The main goal of this paper is on a robust generalized CRT for two integers.

For the case of more than one integer estimation from their residue sets, i.e., the generalized CRT, the reconstruction is more complicated. As mentioned in \cite{xia1999}, the main difficulty for the case of no less than two integers comes from the fact that the correspondence between the original integer and its remainder is not known, which happens when the remainders are obtained by detecting the peaks of the discrete Fourier transforms (DFT) of an undersampled waveform as described in \cite{xia1999}. Moreover, the number of the remainders in a residue set may be less than the number of the integers to determine, since there may be two or more integers sharing the same remainder for some moduli. While all the distinct elements in a residue set are known, the number of repetitions of any remainder is not known in general unless there are only two integers to determine. As mentioned earlier, for the robustness of reconstructing a single integer from its erroneous remainders, it is critical to have a gcd larger than $1$ among all the moduli. This has to hold for the above generalized CRT for multiple integers. However, the generalized CRT methods studied before are only when all the moduli are pairwise co-prime. Therefore, in order to study a robust generalized CRT, we first need to study the generalized CRT when all the moduli have a gcd larger than $1$ and all the remainders are error free. A basic problem then is to determine the dynamic range for a given set of moduli, i.e., the largest range within which multiple nonnegative integers can be uniquely determined from their residue sets modulo the given moduli. For this problem and when all the moduli are pairwise co-prime, several lower bounds for the dynamic range were obtained in \cite{zhou1997}-\cite{xiao2014}. A most recent tight bound was obtained in \cite{wangwei2014} for two integers where a closed-form and a simple determination algorithm were also obtained.

In this paper, we first present the largest dynamic range for two integers when all the moduli have a gcd larger than $1$ and the remaining integers factorized by the gcd of the moduli are pairwise co-prime. For the generalized CRT with erroneous remainders, we obtain a remainder error bound of the eighth of the gcd of all the moduli that leads to a robust estimation of two integers. An efficient reconstruction algorithm is also presented when two integers are within the largest dynamic range. Note that, for the robustness, the remainder error bound, the eighth of the gcd for two integers, seems not surprising,  when the remainder error bound, the quarter of the gcd, for a single integer in CRT is known. However, as we shall see later, the  proof is not trivial at all.

This paper is organized as follows. In Section \ref{second}, we briefly describe the mathematical problem and introduce some notations. In Section \ref{Third}, we present the largest dynamic range and a closed-form determination algorithm for two integers from their error free residue sets, where the moduli are no longer pairwise co-prime. In Section \ref{Fourth}, we present a robust generalized CRT for two integers. In Section \ref{Fifth}, we present an application of the proposed robust generalized CRT in frequency estimation from multiple undersampled waveforms. In Section \ref{Sixth}, we conclude this paper.

\section{Problem Description}\label{second}
We begin with the multiple frequency determination problem
from multiple undersampled waveforms \cite{xia1999}. For simplicity, a complex-valued waveform is given as
\begin{equation}\label{signal}
x(t)=\sum_{l=1}^{L}A_le^{2\pi j f_l t}+ w(t),
\end{equation}
where $w(t)$ is the additive noise, $A_l$ and $f_l$ are nonzero coefficients and frequencies, respectively. Suppose that these frequencies are distinct non-negative integers, i.e., $f_l=N_l$, where $N_l \in \mathbb N$ and $\mathbb N$ denotes the set of natural numbers, $N_i\neq N_j$ for $i\neq j$, in Hz. Let $K\ge2$ and $m_1,\ldots,m_K$ be $K$ positive integers with $1<m_1<\cdots < m_K$. For each $ k \in \{1,\ldots,K\}$, the sampled signal with sampling frequency $m_k$ Hz is
\begin{equation}\label{sample}
x_{m_k}[n]= x \left(\frac{n}{m_k}\right)=\sum_{l=1}^{L}A_le^{2\pi j N_ln/m_k}+ w\left(\frac{n}{m_k}\right),\ n\in \mathbb Z,
\end{equation}
where $\mathbb Z$ denotes the set of integers.
Then, we take the $m_k$-point DFT to $x_{m_k}[n]$ in (\ref{sample}), and obtain
\begin{equation}\label{DFT}
\text{DFT}_{m_k}\big(x_{m_k}[n]\big)[r] = \sum_{l=1}^{L}A_l\delta(r-r_{l,k})+W[r].
\end{equation}
Without considering the influence of noise, remainders $r_{l,k} \equiv N_l \mod m_k$ can be detected from the $m_k$-point DFT without the order information. Then, we have the $K$ error-free residue sets
\begin{equation}\label{S_Correct}
R_k\left(N_1,\ldots,N_L\right)=\bigcup_{l=1}^{L}\{r_{l,k}\}, \ k=1,\ldots,K
\end{equation}
from the $K$ DFTs. In practice, signals are usually corrupted by noises and thus the obtained remainders $r_{l,k}$ may have errors. Let the erroneous remainders be $\tilde r_{l,k}$:
\begin{equation}\label{eq:define_ri}
\tilde r_{l,k} = r_{l,k}+\Delta r_{l,k}, \ l=1,\ldots, L; \ k=1,\ldots,K,
\end{equation}
where $\Delta r_{l,k}$ denote the errors. Then, the erroneous residue sets are
\begin{equation}\label{tilde_S}
\tilde R_k\left(N_1,\ldots,N_L\right)=\bigcup_{l=1}^{L}\{\tilde r_{l,k}\}, \ k=1,\ldots,K.
\end{equation}
The problem is to determine the $L$ frequencies $\{N_1, \ldots, N_L\}$ from these erroneous residue sets.

Under the condition of all the remainders are error-free, $L=2$, and all the $K$ moduli $m_1,\ldots,m_K$ are pairwise co-prime, in \cite{wangwei2014} we obtained the largest dynamic range within which two frequencies (integers), $\{N_1, N_2\}$, can be uniquely determined from their residue sets $R_k(N_1,N_2)$,  where an efficient reconstruction algorithm was also proposed. In this paper, we first generalize the largest dynamic range result obtained in \cite{wangwei2014} from pairwise co-prime moduli $\mathcal M'=\{m_1,\ldots, m_K\}$ to non-pairwise co-prime moduli $\mathcal M=\{M_1, \ldots, M_K\}$ with $M_k=Mm_k$ for $k=1,\ldots,K$, where $0 < m_1 < \cdots <  m_K$ are pairwise co-prime moduli and $M$ is a positive integer. We then study the reconstruction problem of two integers from the erroneous residue sets $\tilde R_1(N_1, N_2), \ldots, \tilde R_K(N_1,N_2)$ modulo $M_k$ for $k=1,\ldots, K$. This question has two parts: 1) the bound of errors, i.e., to what extent of errors we can have a robust estimation of $\{N_1, N_2\}$? 2) how to efficiently and robustly reconstruct $\{N_1, N_2\}$? In what follows, we always denote $\mathcal {M'} =\left\{m_1,\ldots, m_K\right\}$ a set of moduli, $\Gamma = \prod_{k=1}^{K}m_k$, and $\mathcal {M} =\left\{M_1,\ldots, M_K\right\}$ a set of moduli.

\section{Generalized CRT For Two Integers With Error-free Residue Sets}\label{Third}
In this section, we first recall the basics of dynamic range with modulus set $\mathcal M'$ obtained in \cite{wangwei2014}. Then we obtain the largest dynamic range with a modulus set $\mathcal M$ and an efficient method to determine two integers from error-free residue sets.

We first introduce some notations. The remainder of $x$ modulo $y$ is denoted as $\langle x \rangle_y$. For integer $n > 0$, let $\mathbb Z_n$ denote the set $\left\{0,1,\ldots, n-1\right\}$. A set of $n$ elements is called an $n$-set. If we let $\mathcal A =
\{a_1,\ldots, a_L\}$ with $a_l\in \mathbb N$, $l=1,\ldots, L$, $R_k(a_1,\ldots, a_L)$ is also denoted by $R_k(\mathcal A)$.

\begin{definition}\label{defrange}
The dynamic range of a modulus set $\mathcal N =\{n_1, \ldots, n_K\}$ is the minimal positive integer $D$ such that there are two different $L$-sets $\mathcal A$ and $\mathcal B$ with $\mathcal A, \mathcal B \subseteq \mathbb Z_{D+1}$ satisfying $R_k(\mathcal A)=R_k(\mathcal B)$ for each modulus $n_k$. It is denoted by $D_L(n_1,\ldots,n_K)$, or simply $D_L(\mathcal N)$.
\end{definition}

According to Definition \ref{defrange}, if any set of $L$ integers in $\mathbb Z_{D'}$ can be uniquely determined by their remainders modulo $n_1, \ldots, n_K$, then we have $D_{L}(\mathcal N) \ge D'$. On the other hand, if $L$ integers are in $\mathbb Z_{D_L(\mathcal N)}$, then they can be uniquely determined from their remainders modulo $n_1, \ldots, n_K$. Hence, the dynamic range $D_L(\mathcal N)$ in Definition \ref{defrange} is the largest dynamic range within which any $L$ integers are uniquely determined by their remainders modulo $n_1, \ldots, n_K$. For $L=2$ and a given modulus set $\mathcal M'$, the largest dynamic range $D_2(\mathcal M')$ is obtained in \cite{wangwei2014} as follows.

\begin{lemma} \label{d_m}
\cite{wangwei2014} If $m_{K-1}\ge 3$, then $D_2(\mathcal M')=d$. In other words, if $\mathcal M' \ne\{2, 2n+1\}$ for any positive integer $n$, then
\begin{equation}
D_2(\mathcal M')=d,
\end{equation}
where
\begin{equation}\label{defi_d}
d=\min_{I\subseteq \{1,\ldots,K\}}\Bigg\{\prod_{i\in I}m_i+\prod_{i\in \overline{I}}m_i\Bigg\}.
\end{equation}
\end{lemma}

As an example, we consider the case of $m_1=3$, $m_2=5$, and $m_3=7$. According to Lemma \ref{d_m}, we know that the largest dynamic is $d=3\times 5 +7 =22$. Next, we determine the largest dynamic range with modulus set $\mathcal {M}$ for two integers, i.e., $D_2(\mathcal M)$.

\subsection{The Largest Dynamic Range for Two Integers with Modulus Set $\mathcal M$}
\begin{theorem}\label{D_M_prime}
If $m_{1}\ge 3$ and $K > 2$, then $D_2(\mathcal M)= Md$.
\end{theorem}
\begin{proof}
According to the definition of $d$ in (\ref{defi_d}), we have
$$
d \le m_1 + \prod_{k=2}^ {K}m_k < m_1\cdots m_K.
$$
Hence, there must exist a non-empty set $I \subseteq \{1,\ldots,K\}$ such that
\begin{equation}\label{sum}
d=\prod_{k\in I}m_k + \prod_{k\in \overline{I}}m_k.
\end{equation}
We denote the two terms in the summation (\ref{sum}) by $d_1$ and $d_2$:
$$
d_1=\prod_{k\in I}m_k, \  d_2=\prod_{k\in \overline{I}}m_k.
$$
Construct two $2$-sets $\mathcal A_0$ and $\mathcal B_0$ as:
$$
\mathcal A_0=\{0, Md\}, \  \mathcal B_0=\{Md_1, Md_2\}.
$$
It is not difficult to find that $R_k(\mathcal A_0)= R_k(\mathcal B_0)$ holds with moduli $M_1,\ldots, M_K$. According to Definition \ref{defrange}, we obtain
\begin{equation}\label{D_le}
D_2(\mathcal M)\le Md.
\end{equation}
Next, we prove that $D_2(\mathcal M) \ge Md$.

By Definition \ref{defrange}, $D_2(\mathcal M)$ is the minimal positive integer such that there are two different $2$-sets $\mathcal A=\{N'_1, N'_2\}\subseteq \mathbb Z_{D_2(\mathcal M)+1}$ and $\mathcal B=\{N_1,N_2\}\subseteq \mathbb Z_{D_2(\mathcal M)+1}$ with $R_k(\mathcal A)= R_k(\mathcal B)$, $k=1,\ldots,K$. Without loss of generality, we assume $N'_1< N'_2$, $N_1< N_2$, and $N'_1\le N_1$. Clearly,
$$
\min\{\mathcal A \cup \mathcal B\}=N'_1.
$$
By Lemma $1$ in \cite{wangwei2014}, we have
$$
N'_1=0.
$$
Then we have two cases below.

\noindent \textbf{Case 1}: $N_1+N_2\le M\Gamma.$

Since $N'_1=0$ and $R_k(\mathcal A)= R_k(\mathcal B)$ for $k=1,\ldots,K$, we have
$
0\in R_k(\mathcal B).
$
Hence, $M_k$ divides $N_1$ or $N_2$, i.e.,
$$
M_k|N_1 \ \text{or} \ M_k|N_2, \ k=1,\ldots, K.
$$
Define $I=\big\{k\colon M_k|N_1\big\}$ and $J=\big\{k\colon M_k|N_2\big\}$. Then we have
$$
I \cup J = \{1,\ldots, K\},
$$
which means $\overline{I}\subseteq J$. Hence,
$$
N_1\ge M \prod_{i\in I}m_i, \; N_2\ge  M \prod_{i\in J} m_i.
$$
Therefore,
$$
N_1+N_2  \ge   M \prod_{i\in I}m_i+ M \prod_{i\in J}m_i
 \ge  M \prod_{i\in I}m_i+ M \prod_{i\in \overline{I}}m_i
\ge Md.
$$
Note that $\mathcal A =\{N'_1, N'_2\}$ and $\mathcal B = \{N_1,N_2\}$ have the same residue sets, for each modulus $M_k$, we have
$$
N'_1 + N'_2\equiv N_1+N_2 \bmod M_k.
$$
Hence,
$$
N'_2\equiv N_1+N_2\bmod M_k.
$$
That is,
$$
N'_2= N_1+N_2+k M\Gamma \ \text{for some integer $k$}.
$$
If $k\le -1$, then we obtain from $N_1+N_2\le M\Gamma$ that $N'_2\le 0$, which is a contradiction. Therefore, $k\ge 0$, and hence
$$
N'_2\ge N_1+N_2.
$$
It follows from (\ref{addn1n2}) that $N'_2\ge Md$, which leads to
$$
D_2(\mathcal M)\ge N'_2 \ge Md.
$$

\noindent\textbf{Case 2}: $N_1+N_2\ge M\Gamma +1$.

Since $N_2 > N_1$, we obtain
$$
2N_2 \ge N_1+1+N_2 \ge M \Gamma +2.
$$
Note that $D_2(\mathcal M)\ge N_2$. Then we have
$$
D_2(\mathcal M)\ge \bigg\lceil\frac{1}{2}{M\Gamma}\bigg\rceil + 1.
$$
Since $m_1\ge 3$ and $K > 2$, we have
\begin{eqnarray*}
\frac{1}{2}{M\Gamma} & \ge   \frac{1}{2}M m_2\cdots m_K+ M m_2\cdots m_K
\ge   Mm_1+ M m_2\cdots m_K
\ge   Md.
\end{eqnarray*}
Thus,
\begin{equation}\label{D_ge}
D_2(\mathcal M)\ge \min \bigg\{Md, \Big\lceil\frac{1}{2}{M\Gamma}\Big\rceil+1\bigg\} = Md.
\end{equation}
Combining (\ref{D_le}) and (\ref{D_ge}), we obtain
$$
D_2(\mathcal M) = Md.
$$
\end{proof}

\subsection{A Generalized CRT for Two Integers with Modulus Set $\mathcal M$}

We begin with the reconstruction of one integer $N$ with modulus set $\mathcal M$. Let $N$ be an integer to be reconstructed, and $r_k$ be the remainders of $N$ modulo $M_k$, i.e.,
\begin{equation} \label{eq:N}
r_k \equiv N \bmod M_k, \ k=1,\ldots, K,
\end{equation}
where $0\le r_k < M_k$. From (\ref{eq:N}), we have
\begin{equation}\label{rc_modM}
r_k \equiv N \bmod M,\ k=1,\ldots,K.
\end{equation}
That is, all remainders $r_k$ modulo $M$ have the same value, named common remainder \cite{ore}, denoted as $r^c$.
It follows from (\ref{eq:N}) that both $r_k-r^c$ and $N-r^c$ have the same factor $M$.
Let
\begin{equation}\label{N0_q}
Q = (N-r^c)/ {M}
\end{equation}
and
\begin{equation}\label{def_q}
q_k  = (r_k-r^c)/ {M}.
\end{equation}
Then, congruence (\ref{eq:N}) is equivalent to
$$
q_k\equiv Q \bmod m_k,\ \ k=1,\ldots,K.
$$
According to the traditional CRT, $Q$ can be uniquely reconstructed as
\begin{equation} \label{eq:calculate_N_0}
Q \equiv \sum\limits_{k = 1}^ K \Gamma_k \overline \Gamma_k q_k \bmod \Gamma,
\end{equation}
if and only if $Q < \Gamma$, where $\Gamma_k=\Gamma /m_k$, and $\overline \Gamma_k$ is the multiplicative inverse of $\Gamma_k$ modulo $m_k$, i.e.,
$$
\Gamma_k \overline \Gamma_k \equiv 1  \bmod m_k.
$$
Therefore, $N$ can be uniquely reconstructed by
\begin{equation} \label{eq:calculate_N}
N = MQ+r^c.
\end{equation}

Now, we consider the reconstruction of two integers $\{N_1, N_2 \}$ from their error-free residue sets $R_k(N_1,N_2)$ with modulus set $\mathcal M$. Similar to the reconstruction of one integer, the common remainders are significant to the reconstruction. First, from the residue sets, obtain the two common remainders modulo all the remainders by $M$.

Let $\{r^c_1, r^c_2\}$ be the two common remainders. When the two common remainders are not equal, i.e., $r^c_1 \ne r^c_2$, we have $R_k(N_1,N_2)=\left\{ r_{1,k}, r_{2,k} \right\}$ with $r_{1,k}\ne r_{2,k}$. Note that
$$
\{r_1^c, r_2^c\}= \big\{\langle r_{1,k}\rangle_M, \langle r_{2,k}\rangle_M \big\}
$$
holds for each $k$, $k=1,\ldots, K$. On the other hand,
$$
\{r_1^c, r_2^c\}= \big\{\langle N_1 \rangle_M, \langle N_2 \rangle_M \big\}.
$$
Hence, all the remainders in $R_k(N_1,N_2)$ can be split into two sets, $\{ r_{1,1}, \ldots,  r_{1,K}\}$ and $\{ r_{2,1}, \ldots,  r_{2,K}\}$, according to $r_1^c$ and $r_2^c$. Using the traditional CRT, $N_1$ and $N_2$ can be uniquely determined by their remainders $\{ r_{1,1}, \ldots,  r_{1,K}\}$ and $\{ r_{2,1}, \ldots,  r_{2,K}\}$, respectively. This also means that $\{N_1, N_2\}$ can be uniquely determined if and only if $0 \le N_1, N_2 < M \Gamma$.

When the two common remainders are the same, i.e., $r^c_1 = r^c_2 = r^c$, we let
$
q_{l,k}= (r_{l,k}-r^c)/M, \ l=1,2; \ k=1,\ldots, K.
$
Then, (\ref{eq:N}) is equivalent to
\begin{align}\label{equation_star}
q_{l,k} \equiv \frac{N_l -r^c}{M} \bmod m_k.
\end{align}
Denote $R_k(Q_1, Q_2) = \{q_{1,k}, q_{2,k}\}$ for $k=1,\ldots, K$, where $Q_l = (N_l -r^c)/M$ for $l=1,2$. Since $N_l < Md$, we have $Q_l < d$. By the definition of dynamic range, we know that $\{Q_1, Q_2\}$ can be uniquely determined by their residue sets $R_k(Q_1, Q_2)$. Consequently, $\{N_1, N_2\}$ can be uniquely reconstructed by using formula (\ref{eq:calculate_N}).

In summary, we have the following corollary.
\begin{corollary}
Assume that $m_{1}\ge 3$ and $K > 2$. Let $\{r_1^c, r_2^c\}$ be the common remainders defined as above. We have the following results.

\noindent 1) If $r^c_1 \ne r^c_2$ and $0\le N_1, N_2 < M\Gamma$, then $\{N_1, N_2\}$ can be uniquely determined from the above algorithm;

\noindent 2) If $r^c_1 =  r^c_2$ and $0\le N_1, N_2 < Md$, then $\{N_1, N_2\}$ can be uniquely determined from the above algorithm.
\end{corollary}

\textbf{Example 1.} Let $M=100$ and moduli be $m_1=3$, $m_2=5$, and $m_3=7$. By Theorem \ref{D_M_prime}, we obtain that the largest dynamic range is $Md=2200$. Hence, two integers $\{N_1, N_2 \}$ less than $2200$ can be uniquely determined from their residue sets $R_k(N_1, N_2)$. Suppose that the residue sets are $R_1( N_1,  N_2)=\{69, 195\}$, $R_2( N_1,  N_2)=\{95, 169\}$, and $R_3( N_1, N_2)=\{69, 395\}$. Then, the two common remainders are $\{r_1^c, r_2^c\}=\{69, 95\}$. Hence, the remainder in the residue sequences can be split into $\{69, 169, 69\}$ and $\{195, 95, 395\}$ corresponding to $r_1^c=69$ and $r_2^c=95$, respectively. By using the traditional CRT, we have $\{N_1, N_2 \}=\{ 2169, 1095\}$.

\textbf{Example 2.}
Consider the example above. Suppose that the residue sets are $R_1( N_1,  N_2)=\{98, 198\}$, $R_2( N_1,  N_2)=\{98, 398\}$, and $ R_3( N_1, N_2)=\{398, 498\}$ modulo $300$, $500$, and $700$, respectively. In this case, the two common remainders are the same: $r_1^c=r_2^c= 98$. By (\ref{def_q}), we have $R_1(Q_1, Q_2)=\{0, 1\}$, $R_2(Q_1, Q_2)=\{0, 3\}$, and $R_3(Q_1, Q_2)=\{3, 4\}$ modulo $3$, $5$, and $7$, respectively. By the reconstruction algorithm obtained in \cite{wangwei2014}, we have $\{Q_{1},Q_2\}= \{10, 18\}$. By (\ref{eq:calculate_N}), we can reconstruct the two integers $\{N_1, N_2\}$ as $\{1098, 1898\}$.

\section{A Robust Generalized CRT for Two Integers}\label{Fourth}
In this section, we discuss a robust generalized CRT for two integers when the residue sets have errors.

\subsection{Remainders with Errors}
As discussed above, the two common remainders, $\{r_1^c, r_2^c\}$, are the key of the reconstruction of integers $\{N_1, N_2\}$. When remainders are error-free, the two common remainders can be directly determined by any residue set of $\{N_1, N_2\}$. However, this may not be true when residue sets have errors. Take Example $2$ for example. Suppose that the erroneous residue sets are $\tilde R_1( N_1,  N_2)=\{108, 209\}$, $\tilde R_2( N_1,  N_2)=\{92, 399\}$, and $\tilde R_3( N_1, N_2)=\{397, 507\}$. Then, the residue sets modulo $M$ are $\{8,9\}$, $\{92,99\}$, and $\{7,97\}$, respectively. Clearly, the erroneous residue sets $\tilde R_k( N_1,  N_2)$ modulo $M$ are different from each other and we can not directly determine the common remainders $\{r_1^c, r_2^c\}$ from $\tilde R_k( N_1,  N_2)$.

Let $\tilde r_{l,k}^c$ be the remainder of $\tilde r_{l,k}$ modulo $M$, i.e.,
\begin{equation}\label{r_tk}
\tilde r_{l,k}^c = \langle \tilde r_{l,k}\rangle_M, \ l=1,2; \ k=1,\ldots, K.
\end{equation}
In case $\tilde R_k(N_1,N_2)$ has only one element, i.e., $\tilde r_{1,k}=\tilde r_{2,k}$, we have $\tilde r_{1,k}^ c =\tilde r_{2,k}^c$ counted twice (repeated once) in the above sequence. This provides total $2K$ common remainders and some of them may be the same. In order to estimate two common remainders from these $2K$ common remainders $\tilde r_{1,1}^c, \ldots, \tilde r_{1,K}^c$, $\tilde r_{2,1}^c, \ldots, \tilde r_{2,K}^c$, two appropriate clusters, each of which contains $K$ remainders, are formed first. Intuitively the deviation of two clusters should be large. Now, we determine two clusters from these erroneous residue sets. For convenience, we denote these $2K$ common remainders as $\tilde r_{1}^ c, \ldots, \tilde r_{2K}^ c$ and then sort them  in the increasing order as follows
\begin{equation}\label{sequence_rc}
\tilde r^c_{\varsigma_{(1)}}\le\cdots\le \tilde r^c_{\varsigma_{(2K)}},
\end{equation}
where $\varsigma$ is a permutation of the set $\{1,\ldots,2K\}$.

For any two adjacent common remainders $\tilde r^c_{\varsigma_{(k)}}$ and $\tilde r^c_{\varsigma_{(k+1)}}$, we define the distance $D_k$ as
\begin{equation}\label{dis_d}
D_k= \begin{cases}\tilde r^c_{\varsigma_{(k+1)}}-\tilde r^c_{\varsigma_{(k)}},& \text{if} \ k=1,\ldots, 2K-1 \\
\tilde r^c_{\varsigma_{(1)}}-\tilde r^c_{\varsigma_{(2K)}}+M, & \text{if} \ k=2K.
\end{cases}
\end{equation}
It is clear that the nonnegative distances $D_k$ satisfy the following equation
\begin{equation}\label{sum_dk}
\sum_{k=1}^{2K}D_k=M.
\end{equation}
Moreover, we have the following results.

\begin{lemma}\label{d_kM8}
Let $\tau = \text{max}\left\{\left |\Delta r_{l,k}\right |, \ l=1,2; \ k=1,\ldots, K \right \}$, where $\Delta r_{l,k}$ are the remainder errors as defined in (\ref{eq:define_ri}). If $\tau < M/8$, then there exists one and only one subscript $k_0\in \{1,\ldots, K\}$ such that
\begin{equation}\label{d_k}
D_{k_0} + D_{k_0+K} > M/2.
\end{equation}
Moreover, if we let
\begin{eqnarray}\label{Omega_1_2}
\begin{split}
& \Omega_1 \triangleq \{\omega_1,\ldots,\omega_K\}=\left\{\tilde r_{\varsigma_{(k_0+1)}}^c, \ldots, \tilde r_{\varsigma_{(k_0+K)}}^c\right\},
\\
& \Omega_2 \triangleq  \{\upsilon_1,\ldots, \upsilon_K\} =
\begin{cases}
\left\{\tilde r_{\varsigma_{(1)}}^c,\ldots, \tilde r_{\varsigma_{(K)}}^c\right\}, & \text{if} \ k_0=K \\
\left\{\tilde r_{\varsigma_{(k_0+1+K)}}^c- M,\ldots,\tilde r_{\varsigma_{(2K)}}^c- M, \tilde r_{\varsigma_{(1)}}^c, \ldots, \tilde r_{\varsigma_{(k_0)}}^c\right\}, & \text{if} \ k_0 \ne K,
\end{cases}
\end{split}
\end{eqnarray}
with $\omega_i \le \omega_j$, $\upsilon_i \le \upsilon_j$ for $1\le i <j\le K$, then we have
\begin{equation}\label{d_Omega_t}
\omega_K -\omega_1 \le 2 \tau, \ \upsilon_K - \upsilon_1 \le 2 \tau.
\end{equation}
\end{lemma}

This Lemma is proved in Appendix $A$.

\textbf{Example 3.} Let us consider the example proposed at the beginning of this section. By (\ref{sequence_rc}), we obtain the remainder sequence $ \{\tilde r_1^c,\ldots, \tilde r_{2K}^c\}=\{      8,9,92,99,97,7\}$ and its sorted sequence $\{\tilde r^c_{\varsigma_{(1)}}, \ldots, \tilde r^c_{\varsigma_{(2K)}}\}$ in (\ref{sequence_rc}) as
$$
0 < 7 < 8 < 9 < 92 < 97 < 99 < M=100.
$$
Since $D_3+D_6 =(92-9)+(7-99+100)= 83 + 8 > M/2$, we know that $k_0=3$ and obtain from (\ref{Omega_1_2}) that the two clusters are
\begin{equation}\label{Ex3_Omega}
\Omega_1=\{92, 97, 99\}, \ \Omega_2=\{7,8,9\}.
\end{equation}

Before getting the properties of the two clusters $\Omega_1$ and $\Omega_2$, we introduce a kind of circular distance below.
\begin{definition}
For real numbers $x$ and $y$, the circular distance of $x$ to $y$ for a non-zero positive number $C$ is defined as
\begin{equation} \label{eq:define_remainder_distance}
d_C(x,y) \buildrel \Delta \over =  x - y - \left [ \frac {x-y} {C}\right ] C,
\end{equation}
where $\left[\cdot \right]$ stands for the rounding integer,
i.e., for any $x \in \mathbb R$, where $\mathbb R$ denotes the set of all reals, $[x]$ is an integer and subject to
\begin{equation} \label{eq:define[.]}
-  {1}/ {2} \le x- [x]  <  {1} /{2}.
\end{equation}
\end{definition}

\begin{corollary}\label{Omega_M4}
Let $\tau = \text{max}\left\{\left |\Delta r_{l,k}\right |, \ l=1,2; \ k=1,\ldots, K \right \}$ and $\tau < M/8$. If $M/4 \le \left| d_M(r_1^c, r_2^c)\right|\le M/2$, then for every $k\in \{1,\ldots,K\}$, there exist $\omega_{k_1}$ in $\Omega_1$ and $\upsilon_{k_2}$ in $\Omega_2$ such that either
\begin{equation}\label{condi1}
d_M(\tilde r_{1,k}^c, \omega_{k_1})=0 \ \text{and} \ d_M(\tilde r_{2,k}^c, \upsilon_{k_2})=0
\end{equation}
or
\begin{equation}\label{condi2}
d_M(\tilde r_{2,k}^c, \omega_{k_1})=0 \ \text{and} \ d_M( \tilde r_{1,k}^c, \upsilon_{k_2})=0,
\end{equation}
where $k_1,k_2 \in \{1,\ldots,K\}$, $\Omega_1$ and $\Omega_2$ are defined in (\ref{Omega_1_2}).
\end{corollary}

The proof of this corollary is in Appendix $B$.

Based on the two clusters $\Omega_1$ and $\Omega_2$, we can estimate the two common remainders $\left\{r_1^c,r_2^c\right\}$ firstly. Let
\begin{equation}\label{omega_prime}
\omega_k ^\prime =
\begin{cases}
\omega_k, & \text{if} \ \omega_K -\upsilon_1 \le M/2 \\
\omega_k-M, & \text{if} \ \omega_K -\upsilon_1 > M/2,
\end{cases}
\end{equation}
for all $k$, where $k\in \{1,\ldots, K \}$. Then, the two common remainders $\left\{r_1^c,r_2^c\right\}$ can be estimated as $\{\overline \omega_1, \overline \omega_2\}$:
\begin{equation}\label{overline_omega}
\overline \omega_1\triangleq \frac {\omega_1^\prime+\cdots+\omega_K^\prime}{K}, \ \overline \omega_2\triangleq \frac {\upsilon_1+\cdots+\upsilon_K}{K}.
\end{equation}
Note that $\overline \omega_1$ and $\overline \omega_2$ defined in (\ref{overline_omega}) may be negative values. After cancelling the appropriate estimate of common remainder from the erroneous remainders
$\tilde r_{l,k}$, we can obtain the estimates of integers $q_{l,k}$ in (\ref{equation_star}), denoted as $\hat q_{l,k}$:
\begin{equation}\label{calcu_q_k}
\hat q_{l,k}=\left[\frac {\tilde r_{l,k}-\overline \omega_t}{M}\right], \ l=1,2; \ k=1,\ldots, K,
\end{equation}
where $\left[\cdot \right]$ is the rounding operation, and $t$ is
\begin{eqnarray}\label{def_t}
t=
\begin{cases}
1, & \text{if} \ d_M(\tilde r_{l,k}^c, \omega_{k_1})=0 \ \text{for some} \ k_1 \\
2, & \text{if} \ d_M(\tilde r_{l,k}^c, \upsilon_{k_2})=0\ \text{for some} \ k_2,
\end{cases}
\end{eqnarray}
with $k_1,k_2 \in \{1,\ldots, K \}$. Let
\begin{equation}\label{hat_Q}
 R_k(\hat Q_1, \hat Q_2)=\left\{\hat q_{1,k}, \hat q_{2,k}\right\}, \ k=1,\ldots, K.
\end{equation}
Then, the two estimates $\{\hat Q_1, \hat Q_2 \}$ of the integers $\{Q_1, Q_2\}$ can be reconstructed from their residue sets $R_1(\hat Q_1, \hat Q_2)$, $\ldots, R_K(\hat Q_1, \hat Q_2)$ modulo $\mathcal M'$ by using the generalized CRT for two integers obtained in \cite{wangwei2014}.

\textbf{Example 4.} Let us consider Example 3. Since $\omega_3 - \upsilon_1 =99-7= 92 > M/2$, we obtain
$$
\omega_1^\prime = -8, \ \omega_2 ^\prime = -3, \ \omega_3 ^\prime = -1.
$$
Recall that $\upsilon_1=7$, $\upsilon_2=8$, and $\upsilon_3=9$. According to the definitions of $\overline \omega_1$ and $\overline \omega_2$ in (\ref{overline_omega}), we have
$$
\overline \omega_1 = -4, \ \overline\omega_2= 8.
$$
By (\ref{calcu_q_k}) and (\ref{hat_Q}), we obtain
$$
R_1(\hat Q_1, \hat Q_2)=\left\{1, 2\right\}, R_2(\hat Q_1, \hat Q_2)=\left\{1, 4\right\}, R_3(\hat Q_1, \hat Q_2)=\left\{4,5\right\}.
$$
By using the generalized CRT for two integers obtained in \cite{wangwei2014}, we have
$$
\{\hat Q_1, \hat Q_2\}=\{11,19\}.
$$

%
%


Now, we estimate the two integers $\{N_1, N_2\}$ after the estimates $\{\overline \omega_1,\overline \omega_2\}$ of the two common remainders and $\{\hat Q_1, \hat Q_2\}$ are obtained. The estimates of $\{N_1,N_2\}$ are denoted as $\{\hat N_1, \hat N_2\}$ in the following.

\noindent 1) $\hat Q_1= \hat Q_2=\hat Q$.

In this case, the estimates $\{\hat N_1, \hat N_2\}$ can be reconstructed as
\begin{equation}\label{estimate_NAB_eq}
\{\hat N_1, \hat N_2\} =  \{M \hat Q + \overline \omega_1, M \hat Q + \overline \omega_2\}.
\end{equation}

\noindent 2) $\hat Q_1 \ne \hat Q_2$.

In this case, we can not determine $\{\hat N_1, \hat N_2\}$ from $\{\overline \omega_1,\overline \omega_2\}$ and $\{\hat Q_1, \hat Q_2\}$, which is because the correspondence between the elements in two sets $\{\overline \omega_1,\overline \omega_2\}$ and $\{\hat Q_1, \hat Q_2\}$ is not known. To be specific, we cannot determine whether $\{\hat N_1, \hat N_2\}$ are $\{M\hat Q_1 + \overline \omega_1, M\hat Q_2 + \overline \omega_2\}$ or $\{M\hat Q_1 + \overline \omega_2, M\hat Q_2 + \overline \omega_1\}$. Next, we modify the two estimates $\{\overline \omega_1, \overline \omega_2\}$ of the common remainders $\{r_1^c,r_2^c\}$ so that the modified estimates $\hat r_1^c$ and $\hat r_2^c$ correspond to $\hat Q_1$ and $\hat Q_2$, respectively. The main processes are two: Firstly, we select the elements from $\{\omega'_1, \ldots, \omega'_K\}$ and $\{\upsilon_1,\ldots, \upsilon_K\}$ to form two groups, where all the elements in one group correspond to $\hat Q_1$ and the other correspond to $\hat Q_2$. Then, $\hat r_1^c$ and $\hat r_2^c$ are determined by averaging the groups corresponding to $\hat Q_1$ and $\hat Q_2$, respectively.

Recall that the estimates $\{\overline \omega_1, \overline \omega_2\}$ of the two common remainders defined in (\ref{overline_omega}) are the average values of the two clusters $\{\omega'_1, \ldots, \omega'_K\}$ and $\{\upsilon_1,\ldots, \upsilon_K\}$. For convenience, we let
\begin{equation}\label{def_Omega_prime}
\Omega' \triangleq \{\omega'_1, \ldots, \omega'_K, \upsilon_1,\ldots, \upsilon_K\}.
\end{equation}
By (\ref{Omega_1_2}) and (\ref{omega_prime}), we know that the elements in $\Omega'$ are either $\tilde r_{\varsigma_{(i)}}^c$ or $\tilde r_{\varsigma_{(i)}}^c-M$ for all $i=1,\ldots, 2K$. From the definitions of $\tilde r_{\varsigma_{(i)}}^c$ in (\ref{sequence_rc}), we know that $\{\tilde r_{\varsigma_{(1)}}^c, \ldots, \tilde r_{\varsigma_{(2K)}}^c\}$ are the $2K$ sorted common remainders from $\{\tilde r_{1,1}^c, \ldots, \tilde r_{1,K}^c,\tilde r_{2,1}^c, \ldots, \tilde r_{2,K}^c\}$. Hence, for all $l=1,2; k=1,\ldots,K$, either $\tilde r_{l,k}^c$ or $\tilde r_{l,k}^c-M$ is included in $\Omega'$, and in the meanwhile, $\Omega'$ only consists of these $2K$ elements. In the following, for convenience, we call both $\tilde r_{l,k}^c$ and $\tilde r_{l,k}^c-M$ as the common remainders of $\tilde r_{l,k}$.

When $\hat Q_1 \ne \hat Q_2$, we know from the traditional CRT that there exists at least a subscript $k\in\{1,\ldots, K\}$ such that
\begin{equation}\label{q_ne_k}
\hat q_{1,k} \ne \hat q_{2,k}.
\end{equation}
Let $\mathcal K \triangleq \{k_1,\ldots,k_p\}$, $1 \le k_i \le K$, be all the distinct subscripts of $\hat q_{1,k_i}$ (or $\hat q_{2,k_i}$) satisfying (\ref{q_ne_k}), i.e., $\hat q_{1,k_i} \ne \hat q_{2,k_i}$, and thus from (\ref{q_ne_k}), we have $p\ge 1$. When $\hat q_{1,k_i} \ne \hat q_{2,k_i}$,  the correspondence between $\{\hat q_{1,k_i}, \hat q_{2,k_i}\}$ and $\{\hat{Q}_1, \hat{Q}_2\}$ is known because
we can determine $\hat q_{l,k_i}$ by the obtained integers $\hat Q_l$ modulo $m_{k_i}$, i.e.,
\begin{equation}
\hat q_{l,k_i} = \langle \hat Q_l \rangle_{m_{k_i}}, \ l=1,2,
\end{equation}
for every $i$, $1\le i \le p$. Note that the obtained values $\hat q_{1,k_i}$ and $\hat q_{2,k_i}$ above are the same the values as determined by (\ref{calcu_q_k}) from the residue set $\{\tilde r_{1,k_i},\tilde r_{2,k_i}\}$. From $\hat q_{1,k_i} \ne \hat q_{2,k_i}$, we deduce that $\tilde r_{1,k_i} \ne \tilde r_{2,k_i}$. Thus, the correspondence between $\{\hat q_{1,k_i}, \hat q_{2,k_i}\}$ and $\{\tilde r_{1,k_i}, \tilde r_{2,k_i}\}$ (\text{or} $\{\hat{Q}_1, \hat{Q}_2\}$) is known as well. Assume that the common remainders of $\tilde r_{1,k_i}$ and $\tilde r_{2,k_i}$ in $\Omega'$ are $\hat r_{1,k_i}^c$ and $\hat r_{2,k_i}^c$, respectively. As discussed above, $\hat r_{1,k_i}^c$  are either $\tilde r_{1,k_i}^c$ or $\tilde r_{1,k_i}^c-M$, and $\hat r_{2,k_i}^c$  are either $\tilde r_{2,k_i}^c$ or $\tilde r_{2,k_i}^c-M$. Thus, $\hat r_{1,k_i}^c$ and $\hat r_{2,k_i}^c$ can be determined by
\begin{equation}\label{r_sigma_eta}
d_M(\tilde r_{1,k_i}^c, \hat r_{1,k_i}^c)=0, \ d_M(\tilde r_{2,k_i}^c, \hat r_{2,k_i}^c)=0, \ \hat r_{1,k_i}^c, \hat r_{2,k_i}^c \in \Omega',
\end{equation}
where $k_i \in \mathcal K$. By (\ref{r_tk}), we know that $\hat r_{1,k_i}^c$ and $\hat r_{2,k_i}^c$ can also be determined by
\begin{equation}\label{r_sigma_eta_2}
d_M(\tilde r_{1,k_i}, \hat r_{1,k_i}^c)=0, \ d_M(\tilde r_{2,k_i}, \hat r_{2,k_i}^c)=0, \ \hat r_{1,k_i}^c, \hat r_{2,k_i}^c \in \Omega'.
\end{equation}
Clearly, $\hat r_{1,k_i}^c$ and $\hat r_{2,k_i}^c$ correspond to the remainders $\tilde r_{1,k_i}$ and $\tilde r_{2,k_i}$, respectively. Hence, $\hat r_{1,k_i}^c$ corresponds to $\hat Q_1$, while $\hat r_{2,k_i}^c$  corresponds to $\hat Q_2$. We use the average common remainders of $\{\hat r_{1,k_1}^c, \ldots, \hat r_{1,k_p}^c\}$ and $\{\hat r_{2,k_1}^c, \ldots, \hat r_{2,k_p}^c \}$ as the estimates of $r^c_1$ and $r^c_2$, respectively, i.e.,
\begin{equation}\label{rc_AB}
\hat r^c_{1}  \triangleq \frac {\hat r_{1,k_1}^c+\cdots +\hat r_{1,k_p}^c}{p}, \ \hat r^c_{2} \triangleq \frac {\hat r_{2,k_1}^c+\cdots+\hat r_{2,k_p}^c}{p}.
\end{equation}
Then, the estimates $\{\hat N_1, \hat N_2\}$ can be reconstructed as
\begin{equation}\label{estimate_NAB}
\{\hat N_1, \hat N_2\} = \{M \hat Q_1 + \hat r^c_{1}, M \hat Q_2 + \hat r^c_{2}\}.
\end{equation}

Noting that the estimates $\{\hat N_1, \hat N_2\}$ obtained by (\ref{estimate_NAB_eq}) or (\ref{estimate_NAB}) may be non-integers. For this case, we use $\big\{[\hat N_1], [\hat N_2]\big\}$ as the estimates of the integers $\{N_1, N_2\}$, where $[\cdot]$ denotes the rounding operation defined in (\ref{eq:define[.]}).

\textbf{Example 5.} Let us consider Example $4$. Note that $\{\hat Q_1, \hat Q_2\}=\{11,19\}$ calculated before. Then, the remainders of $\hat Q_1=11$ and $\hat Q_2=19$ modulo $\mathcal M'=\{3,5,7\}$ are $\{\hat q_{1,1}, \hat q_{1,2}, \hat q_{1,3}\}=\{2,1,4\}$ and $\{\hat q_{2,1}, \hat q_{2,2}, \hat q_{2,3}\}=\{1,4,5\}$, respectively. Clearly, $\hat q_{1,1} \ne \hat q_{2,1}$, $\hat q_{1,2} \ne \hat q_{2,2}$, and $\hat q_{1,3} \ne \hat q_{2,3}$. Recall that the erroneous residue sets are $\tilde R_1( N_1,  N_2)=\{108, 209\}$, $\tilde R_2( N_1,  N_2)=\{92, 399\}$, and $\tilde R_3( N_1, N_2)=\{397, 507\}$. According to (\ref{calcu_q_k}), we deduce that $\{\tilde r_{1,1},  \tilde r_{1,2}, \tilde r_{1,3}\}=\{209, 92, 397\}$, $\{\tilde r_{2,1}, \tilde r_{2,2},  \tilde r_{2,3}\}=\{108, 399, 507\}$. By (\ref{def_Omega_prime}), we have
$$
\Omega^\prime=\{\omega_1^\prime, \omega_2^\prime, \omega_3^\prime, \upsilon_1, \upsilon_2, \upsilon_3\}=\{-8,-3, -1,7,8,9\}.
$$
Note that
\begin{eqnarray*}
\begin{split}
 d_M(\tilde r_{1,1},9)=0,\
 d_M(\tilde r_{1,2},-8)=0,\
 d_M(\tilde r_{1,3},-3)=0.
\end{split}
\end{eqnarray*}
By (\ref{r_sigma_eta_2}), we obtain that the common remainders $\hat r_{1,1}^c$, $\hat r_{1,2}^c$, and $\hat r_{1,3}^c$ are
$$
\hat r_{1,1}^c=9,\ \hat r_{1,2}^c=-8,\ \hat r_{1,3}^c=-3.
$$
From (\ref{rc_AB}), we obtain
$$
\hat r_{1}^c = -2/3.
$$
Similarly, we have
\begin{eqnarray*}
 d_M(\tilde r_{2,1},8)=0,\
 d_M(\tilde r_{2,2},-1)=0,\
 d_M(\tilde r_{2,3},7)=0.
\end{eqnarray*}
Hence, 
$$
\hat r_{2,1}^c=8,\ \hat r_{2,2}^c=-1,\ \hat r_{2,3}^c=7,
$$
and then we obtain from (\ref{rc_AB}) that
$
\hat r_{2}^c = 14/3.
$
By (\ref{estimate_NAB}), we have
$$
\{\hat N_1, \hat N_2\}=\Big\{1099\frac{1}{3}, 1904\frac{2}{3}\Big\}.
$$
Therefore, the estimates of the two integers $\{N_1,N_2\}$ are $\{1099, 1905\}$. Note that the true values of the two integers are $\{1098,1898\}$.

Next theorem shows that the above estimates $\{\hat N_1, \hat N_2\}$ of the two integers $\{N_1, N_2\}$ are robust when the remainder error bound is less than $M/8$.

\begin{theorem}\label{Th_N_estimate}
Let $\tau = \text{max}\left\{\left |\Delta r_{l,k}\right |,\ l=1,2; \ k=1,\ldots, K \right \}$, where $\Delta r_{l,k}$ are the remainder errors as defined in (\ref{eq:define_ri}). If $\tau < M/8$, then we have
\begin{equation}\label{N_robust}
\big|\hat N_l - N_l \big| \le \tau, \ l=1,2,
\end{equation}
where $\{\hat N_1, \hat N_2\}$ are defined in (\ref{estimate_NAB_eq}) or (\ref{estimate_NAB}).
\end{theorem}
The proof of this theorem is in Appendix $C$.

Let us recall the example presented at the beginning of this section. Note that the remainder error bound $\tau=11$, which is less than the robustness error upper bound $M/8=12.5$. By Theorem \ref{Th_N_estimate}, we know that the estimates are robust. In fact, according to Example $5$, the maximal estimation error of the two integers is $7$, which is small than the remainder error bound $\tau$ and conforms the result obtained in Theorem \ref{Th_N_estimate}.

\subsection{Robust Generalized CRT Algorithm for Two Integers}

To summarize what we have studied before, we obtain the following robust generalized CRT algorithm for two integers.

\noindent \textbf{Robust Generalized CRT for Two Integers}

\noindent\textbf{Step 1} Calculate $\tilde r^c_{l,k}$ in (\ref{r_tk}) and sort them in the increasing order as (\ref{sequence_rc}).

\noindent\textbf{Step 2} Compute $k_0$ as
\begin{equation}
k_0 = \arg \max_{k \in \left\{1,\ldots, K \right\}}\left\{ D_k + D_{k+K} \right\},
\end{equation}
where $D_k$ is defined in (\ref{dis_d}).

\noindent\textbf{Step 3} Obtain the two clusters $\Omega_1$ and $\Omega_2$ by (\ref{Omega_1_2}).

\noindent\textbf{Step 4} Calculate $\overline \omega_1$ and $\overline \omega_2$ by (\ref{overline_omega}).

\noindent\textbf{Step 5} Determine residue sets $R_k(\hat Q_1, \hat Q_2)$ as
\begin{equation}
R_k(\hat Q_1, \hat Q_2)=\left\{\hat q_{1,k}, \hat q_{2,k}\right\},
\end{equation}
where $\hat q_{l,k}$ are defined in (\ref{calcu_q_k}).


\noindent\textbf{Step 6} Reconstruct $\{\hat Q_{1},\hat Q_{2}\}$ by using the generalized CRT for two integers obtained in \cite{wangwei2014}.

\noindent\textbf{Step 7} Reconstruct $\{\hat N_1, \hat N_2\}$ by (\ref{estimate_NAB_eq}) or (\ref{estimate_NAB}).

Although the above robust generalized CRT is for two integers, it is straightforward to be generalized to two reals as the case of one integer in our previous work \cite{wjwang}, \cite{wangwj_14}.

\section{Simulation Results}\label{Fifth}
In this section, we show some simulations to illustrate the performance of the proposed robust generalized CRT for two integers and its application in two frequency determination from multiple undersampled waveforms.

Let us first consider the estimation error versus the error upper bound for the proposed robust generalized CRT for two integers. By Theorem \ref{Th_N_estimate}, we know that the maximal error level $\tau$ needs to be upper bounded by $\tau < M/8$ for the robustness. In the simulation, parameter $M=100$, and the co-prime integers from $m_1$ to $m_3$ are $3$, $5$, and $7$, respectively. Two unknown integers $\{N_1, N_2\}$ are chosen uniformly at random from the interval $[0, 2000)$ and the maximal error levels are set as $\tau=0, 3, 6, 9, 12, 15$. For these maximal error levels, the last one, $15$, does not satisfy the robustness upper bound $\tau< M/8=100/8=12.5$. We call the process of determining $\{N_1, N_2\}$ as a trial, and $10000$ trials for each of the maximal error level are simulated. In Fig. \ref{Bound}, we present the curve of the mean error $E_N$ versus the maximal error level $\tau$. The mean error is defined as
\begin{equation}
E_N = E_{trials}\bigg\{\frac{1}{2} \sum_{l=1,2} |\hat{N}_l -N_l |\bigg\},
\end{equation}
where $E_{trials}$ stands for the mean over all the trials, $N_{l}$ and $\hat N_{l}$ are the true integers and the estimates in one trial, respectively. Fig. \ref{Bound} shows that the two integers can be robustly reconstructed from their erroneous residue sets by using the proposed robust generalized CRT for two integers, i.e., when all the errors of the remainders are less than the error upper bound, the reconstruction errors of $\{N_1, N_2\}$ are also less than this bound. It also shows that the reconstruction errors of the two integers are small compared to their dynamic range. However, when the robustness upper bound $\tau<M/8$ is not satisfied as when $\tau=15$, as one can see from Fig. \ref{Bound}, the robustness may not hold anymore.


For the application in two frequency determination from multiple undersampled waveforms, we set three sampling frequencies: $Mm_1$, $Mm_2$, and $Mm_3$. Two frequencies $\{f_1, f_2\}$ are taken integers randomly and uniformly distributed in the range $\left(0, 2M\sqrt{m_1m_2m_3}\right)$. The noise $w(t)$ in (\ref{signal}) is additive white Gaussian noise with mean zero and variance $10^{-SNR/10}$, and the number of trials is $10000$ for each signal-to-noise ratio (SNR). In the simulation, the observation of the time duration is $1$s. Three methods are considered: the (optimal) searching based method, the proposed robust generalized CRT for two integers and the non-robust generalized CRT for two integers.
In the searching based method, we search the proper folding integers $\hat n_{l,k}$ that corresponds to the remainders $\tilde r_{l,k}$ from all the possible integers. Then the two frequencies are estimated as
\begin{equation}
\hat f_l =\Bigg[\frac{1}{K} \sum_{k=1}^{K}\hat n_{l,k}Mm_k + \tilde r_{l,k} \Bigg], \ l=1,2.
\end{equation}
In non-robust generalized CRT for two integers, we choose the average of the common remainders by using an arbitrary grouping, and then reconstruct the different integers by using the algorithm proposed in Section \ref{Third}. In Fig. \ref{non_robust_R1}, we compare the different methods by investigating the mean relative error $E_f$ versus SNR of the two estimated frequencies for different $M$ and $m_i$. The mean relative error is defined as
\begin{equation}
E_f =  E_{trials}\Bigg\{\frac{1}{2} \sum_{l=1,2} \frac{|\hat{f}_l - f_l |}{f_l}\Bigg\},
\end{equation}
where $f_{l}$ and $\hat f_{l}$ are the true frequencies and the estimates in one trial, respectively. Fig. \ref{non_robust_R1} shows that the non-robust generalized CRT for two integers suffers from the error floor problem, i.e., the mean relative error will not decrease or decrease very slowly at high SNR. On the contrary, the mean relative error of the robust generalized CRT and the searching based method for two integers decreases sharply as SNR increases all the time. The proposed robust generalized CRT performs slightly worse than the searching based method, but has a much less computation. In fact, the computational complexity of the searching method and our proposed method are in the order of $2^K\Gamma^{2(K-1)}$ and $6 K^2$, respectively. Fig. \ref{non_robust_R1} also shows that while the other parameters, such as the sampling rates, are similar, the larger $M$ is, the better reconstruction is, i.e., the better performance is, which is in agreement with our theory.

\begin{figure}[h]
\centering
\includegraphics[width=3.5in]{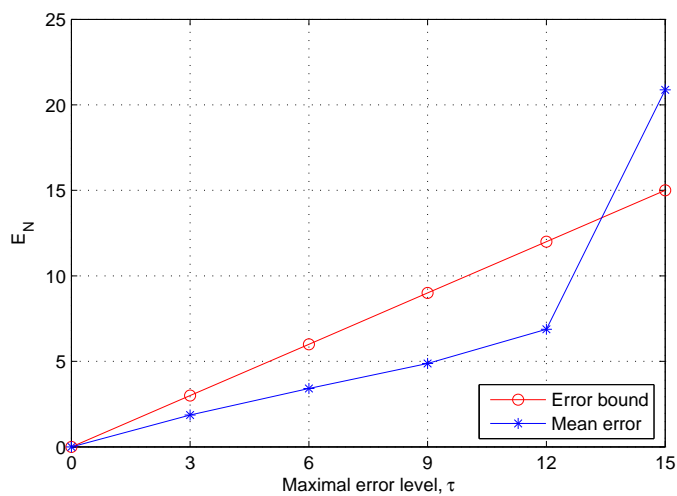}
\captionsetup{font={scriptsize}}
\captionsetup{labelsep=period}
\caption{Estimation errors and the obtained estimation error upper bound using the robust generalized CRT for two integers.} \label{Bound}
\end{figure}

\begin{figure}[h]
\centering
\setlength{\belowcaptionskip}{-.5cm}
\includegraphics[width=3.5in]{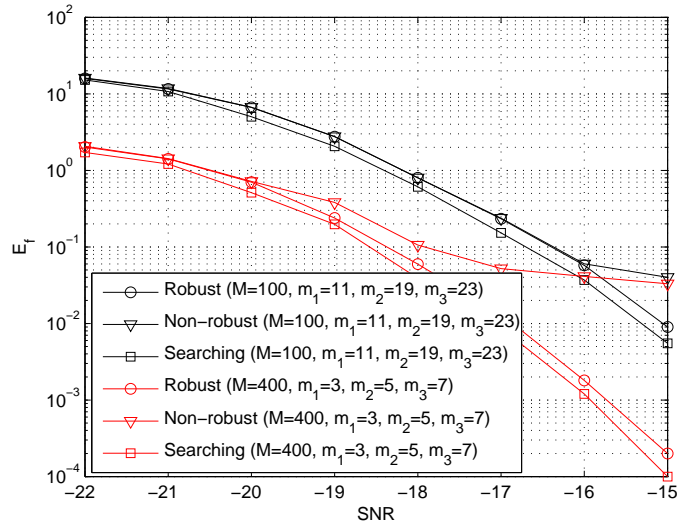}
\captionsetup{font={scriptsize}}
\captionsetup{labelsep=period}
\caption{Mean relative error versus SNR of the two estimated frequencies.} \label{non_robust_R1}
\end{figure}


\section{Conclusion} \label{Sixth}
In this paper, we studied a robust generalized CRT for determining two integers from their residue sets and moduli, where the remainders of the two integers in each residue set are not ordered and may have errors. We first obtained the largest dynamic range of two integers from their error free residue sets of a given modulus set, where all the moduli have a gcd $M$ larger than $1$ and the remaining integers factorized by the gcd of all the moduli are pairwise co-prime. We also presented an efficient reconstruction algorithm of two integers from their error free residue sets, when the two integers are within the largest dynamic range. We then proved that the two integers can be robustly reconstructed if their remainder errors are less than the eighth of the gcd of all the moduli. Finally, we applied the proposed robust generalized CRT for two integers to the determination of two frequencies from multiple undersampled waveforms. Our numerical results showed that the frequency determination performance using our newly proposed robust generalized CRT is better than that using the non-robust generalized CRT. Compared with the optimal searching based method, it has a slightly worse performance but much less computation.

\section*{Appendix}

\subsection{Proof of Lemma \ref{d_kM8}}\label{appendix_1}

\begin{proof}
Without loss of generality, we suppose $r^c_1 \le r_2^c$. Our proof consists of two steps: firstly, we  prove that there exists a subscript $k_0\in\{1,\ldots, K \}$ satisfying (\ref{d_k}). Furthermore, we obtain two clusters, $\Omega_1$ and $\Omega_2$, and prove that they satisfy (\ref{d_Omega_t}). Then, we prove the uniqueness of such $k_0$. By the definition of circular distance in (\ref{eq:define_remainder_distance}), we obtain
\begin{equation}
0 \le \left| d_M(r_1^c, r_2^c) \right| \le M/2.
\end{equation}
Then, we have two cases below.

\noindent \textbf{Case 1:} $0 \le \left| d_M(r_1^c, r_2^c) \right| < M/4$.

In this case, we have $0 \le r^c_2-r^c_1 < M/4$ or $3M/4 < r^c_2-r^c_1 < M$. Let $\rho$ and $\pi$ be two permutations of the set $\left\{1, \ldots, K \right\}$ such that
\begin{equation}\label{def_rho_pi}
\Delta r_{\rho_{(1)}}\le \cdots \le \Delta r_{\rho_{(K)}} \ \text{and} \ \Delta r_{\pi_{(1)}}\le \cdots \le \Delta r_{\pi_{(K)}},
\end{equation}
respectively, where $\Delta r_{\rho_{(k)}}\in \left\{ \Delta r_{1,1}, \ldots, \Delta r_{1,K}\right\}$, $\Delta r_{\pi_{(k)}}\in \left\{ \Delta r_{2,1}, \ldots, \Delta r_{2,K}\right\}$ for $k=1,\ldots,K$.
Define
\begin{eqnarray}\label{set_r1c}
&& \{c_1,\ldots, c_{2K}\} \triangleq  \nonumber \\
&&
\begin{cases}
\left\{r_1^c +\Delta r_{\rho_{(1)}}, \ldots, r_1^c +\Delta r_{\rho_{(K)}}, r_2^c +\Delta r_{\pi_{(1)}},\ldots, r_2^c +\Delta r_{\pi_{(K)}}\right\}, & \hspace{-2mm} \text{if} \ 0 \! \le \! r^c_2 \!- \! r^c_1 \! < \! M/4 \\
\left\{r_1^c +\Delta r_{\rho_{(1)}}, \ldots, r_1^c +\Delta r_{\rho_{(K)}}, r_2^c  +\Delta r_{\pi_{(1)}} \!\! -M,\ldots, r_2^c +\Delta r_{\pi_{(K)}} \!\! -M\right\}, & \hspace{-2mm} \text{if} \ 3M/4\! <\! r^c_2 \!- \!r^c_1 \! <  \!M.
\end{cases}
\end{eqnarray}
where $c_i$ are sorted in the increasing order as
\begin{equation} \label{b_sequence}
c_1 \le \cdots \le c_{2K}.
\end{equation}
Note that $\left|\Delta_{l,k}\right|< M/8$ for $l=1,2; k=1,\ldots, K$. If $0 \le r^c_2-r^c_1 < M/4$, then we have
\begin{equation} \label{rc_difference}
-M/4 < r_2^c + \Delta r_{\pi_{(K)}}- r_1^c - \Delta r_{\rho_{(1)}}< M/2.
\end{equation}
If $3M/4 < r^c_2-r^c_1 < M$, then we have
\begin{equation} \label{rc_difference2}
-M/4 <  r_1^c + \Delta r_{\rho_{(K)}} -(r_2^c+ \Delta r_{\pi_{(1)}}-M) < M/2.
\end{equation}
Therefore,
\begin{equation} \label{b_dis}
0 \le c_{2K}- c_1< M/2.
\end{equation}
Let
\begin{equation} \label{b_presentation}
 c_i =\alpha_i + \ell_i M, \ i=1,\ldots, 2K,
\end{equation}
where $0\le \alpha_i < M$ and $\ell_i \in \mathbb Z$. Then, we obtain from (\ref{b_sequence}) that
\begin{equation}\label{seq_l}
\ell_1 \le \cdots \le \ell_{2K}.
\end{equation}
By (\ref{b_dis}) and (\ref{b_presentation}), we have
\begin{equation}\label{alpha_diff}
0\le \alpha_{2K} -\alpha_1 + (\ell_{2K}-\ell_1)M < M/2.
\end{equation}
Since $0\le \alpha_i < M$, we have $-M < \alpha_{2K} -\alpha_1 < M$. It follows from (\ref{seq_l}) and (\ref{alpha_diff}) that
$$
\ell_{2K}-\ell_1 =0 \ \text{or} \ 1.
$$

\noindent \textbf{Subcase 1:} $\ell_{2K}-\ell_1 =0$.

In this case, $\ell_1=\cdots=\ell_{2K}$.  From (\ref{alpha_diff}),
we have
\begin{equation}\label{alpha_d}
0 \le \alpha_{2K} - \alpha_1 < M/2,
\end{equation}
and from (\ref{b_sequence}) and (\ref{b_presentation}) we have
\begin{equation}\label{case2_1}
0\le \alpha_1 \le \cdots \le \alpha_{2K}< M.
\end{equation}
From (\ref{eq:define_ri}) and (\ref{def_q}), we have
\begin{equation}\label{r_lk_q}
\tilde r_{l,k} = Mq_{l,k} + r_l^c +\Delta r_{l,k}, \ l=1,2; \ k=1,\ldots, K.
\end{equation}
Hence,
\begin{equation}\label{r_lk_M}
\langle \tilde r_{l,k} \rangle_M = \langle  r_l^c +\Delta r_{l,k} \rangle_M, \ l=1,2; \ k=1,\ldots, K.
\end{equation}
On the other hand, we obtain from (\ref{b_presentation}) that
\begin{equation}\label{bi_M}
\langle c_i\rangle_M = \alpha_i, \ i=1,\ldots, 2K.
\end{equation}
Combining  (\ref{set_r1c}), (\ref{r_lk_M}), and (\ref{bi_M}), we have that $\left\{\alpha_1, \ldots, \alpha_{2K}\right\}$ are the $2K$ remainders $\left\{\tilde r_{1}^ c, \ldots, \tilde r_{2K}^ c\right\}$. Hence, (\ref{sequence_rc}) is equivalent to (\ref{case2_1}).
If we let $k_0=K$, then we obtain from (\ref{alpha_d}) that
\begin{eqnarray*}\label{dge0_1}
D_{k_0} + D_{k_0+K} &=& D_{k_0} + \tilde r_{\varsigma_{(1)}}^c -\tilde r_{\varsigma_{(2K)}}^c + M \nonumber \\
&=& D_{k_0} + \alpha_1 - \alpha_{2K} +M \nonumber\\
&>& M/2.
\end{eqnarray*}
By the definitions of $\Omega_1$ and $\Omega_2$, we obtain from (\ref{case2_1}) that
\begin{eqnarray}\label{set2_1}
\begin{split}
& \Omega_1=\left\{\alpha_{K+1}, \ldots, \alpha_{2K}\right\}= \left\{c_{K+1}-\ell_{K+1}M, \ldots, c_{2K}-\ell_{2K}M\right\}, \\
& \Omega_2=\left\{\alpha_{1}, \ldots, \alpha_{K} \right\}= \left\{c_1-\ell_1M, \ldots, c_{K}-\ell_{K}M\right\}.
\end{split}
\end{eqnarray}
Recall that $c_1,\ldots,c_{2K}$ are sorted in the increasing order from erroneous remainders $r_1^c+\Delta r_{\rho_{(1)}}, \ldots, r_1^c+\Delta r_{\rho_{(K)}}, r_2^c+\Delta r_{\pi_{(1)}}, \ldots, r_2^c+\Delta r_{\pi_{(K)}}$. Since $\left|\Delta r_{l,k}\right|\le \tau$ for $l=1,2; k=1,\ldots,K$, we have
\begin{equation}\label{C_gamma}
c_K-c_1\le 2\tau, \  c_{2K}-c_{K+1}\le 2\tau.
\end{equation}
Thus,
\begin{eqnarray*}
&& \omega_K-\omega_1= \alpha_{2K}- \alpha_{K+1} = c_{2K}- c_{K+1} \le 2\tau, \\
&& \upsilon_K-\upsilon_1= \alpha_{K}- \alpha_1 = c_{K}- c_1 \le 2\tau.
\end{eqnarray*}

\noindent \textbf{Subcase 2:} $\ell_{2K}-\ell_1 =1 $.

In this case, there exist some $j\in\{1,\ldots, 2K\}$ satisfying $\ell_{j+1}-\ell_j =1$. Due to (\ref{seq_l}), such subscript $j$ is the only one. Moreover, we have $\ell_1= \cdots =\ell_j$ and $\ell_{j+1}=\cdots=\ell_{2K}$.
From (\ref{b_sequence}), (\ref{b_dis}), and (\ref{b_presentation}), we have
\begin{eqnarray*}
\begin{split}
& \alpha_1 \le \cdots \le \alpha_j, \;  \alpha_j - \alpha_1 < M/2, \\
& \alpha_{j+1} \le \cdots \le \alpha_{2K}, \;  \alpha_{2K} - \alpha_{j+1} < M/2.
\end{split}
\end{eqnarray*}
Since $0 \le c_{2K}- c_1< M/2$ and $\ell_{2K}-\ell_1=1$, we obtain from  (\ref{b_presentation}) that $\alpha_{2K} - \alpha_1 < -M/2$, i.e.,
\begin{equation}\label{alpha_1_gamma}
\alpha_1 - \alpha_{2K} > M/2.
\end{equation}
Thus,
\begin{equation}\label{case2_2}
0 \le \alpha_{j+1} \le \cdots \le \alpha_{2K} < \alpha_1 \le \cdots \le \alpha_j< M.
\end{equation}
Note that $\left\{\alpha_1, \ldots, \alpha_{2K}\right\}$ are the $2K$ remainders $\left\{\tilde r_{1}^ c, \ldots, \tilde r_{2K}^ c\right\}$. Hence, (\ref{sequence_rc}) is equivalent to (\ref{case2_2}).

\noindent 1) If $j<K$ and let $k_0=K-j$, then we obtain from (\ref{alpha_1_gamma}) that
\begin{eqnarray*}\label{dge0_2}
D_{k_0} + D_{k_0+K} & = & D_{k_0}+ \tilde r_{\varsigma_{(2K-j+1)}}^c -\tilde r_{\varsigma_{(2K-j)}}^c \nonumber \\
&=& D_{k_0}+ \alpha_1- \alpha_{2K} \nonumber \\
&>& M/2.
\end{eqnarray*}
By the definitions of $\Omega_1$ and $\Omega_2$, we obtain from (\ref{case2_2}) that
\begin{eqnarray}\label{set2_21}
\begin{split}
&\Omega_1=\left\{\alpha_{K+1}, \ldots, \alpha_{2K}\right\}= \left\{c_{K+1}-\ell _{K+1}M, \ldots, c_{2K}-\ell_{2K}M\right\},  \\
 &\Omega_2 = \left\{\alpha_1-M, \ldots, \alpha_j-M,
\alpha_{j+1}, \ldots, \alpha_{K}\right\}  \\
& \quad =\left\{c_1-M-\ell_1M, \ldots, c_j-M-\ell_jM,
c_{j+1}-\ell_{j+1}M, \ldots, c_{K}-\ell_{K}M\right\}.
\end{split}
\end{eqnarray}
Hence, similar to (\ref{C_gamma}) we have
$$
\omega_K-\omega_1= c_{2K}-c_{K+1} \le 2\tau, \
 \upsilon_K-\upsilon_1 = c_{K}-c_1 \le 2\tau.
$$

\noindent 2) If $j\ge K$ and let $k_0=2K-j$, then we obtain from (\ref{alpha_1_gamma}) that
\begin{eqnarray*}\label{dge0_3}
D_{k_0} + D_{k_0+K} & = & \tilde r_{\varsigma_{(2K-j+1)}}^c -\tilde r_{\varsigma_{(2K-j)}}^c+ D_{k_0+K} \nonumber \\
&=& \alpha_1- \alpha_{2K}+ D_{k_0+K}  \nonumber \\
&>& M/2.
\end{eqnarray*}
By the definitions of $\Omega_1$ and $\Omega_2$, we obtain from (\ref{case2_2}) that
\begin{eqnarray}\label{set2_22}
\begin{split}
&\Omega_1=\left\{\alpha_1, \ldots, \alpha_{K}\right\}= \left\{c_1-\ell_1M, \ldots,
c_{K}-\ell_{K}M\right\},  \\
& \Omega_2 =\left\{\alpha_{K+1}-M, \ldots, \alpha_j-M,
\alpha_{j+1}, \ldots, \alpha_{2K}\right \}  \\
& \quad =\left\{c_{K+1}-M-\ell_{K+1}M, \ldots, c_j-M-\ell_j M,
c_{j+1}-\ell_{j+1}M, \ldots, c_{2K}-\ell_{2K}M\right\}.
\end{split}
\end{eqnarray}
Hence, similar to (\ref{C_gamma}) we have
\begin{eqnarray*}
\begin{split}
& \omega_K -\omega_1= \alpha_{K}- \alpha_1 = c_{K}- c_1 \le 2\tau , \\
&  \upsilon_K-\upsilon_1= \alpha_{2K}- \alpha_{K+1}+M = c_{2K}- c_{K+1}\le 2\tau.
\end{split}
\end{eqnarray*}

\noindent \textbf{Case 2:} $M/4 \le \left| d_M(r_1^c, r_2^c) \right| \le M/2$.

In this case, we have $M/4\le r_2^c - r_1^c \le M/2$ or $M/2 < r_2^c - r_1^c  \le 3M/4$. Hence, $M/4\le r_2^c - r_1^c \le 3M/4$. Since $|\Delta r_{l,k}|< M/8$ for $l=1,2$; $k=1,\ldots, K$, we have
\begin{equation}\label{condi_r1r2}
0 < r_2^c + \Delta r_{\pi_{(k_2)}}-r_1^c - \Delta r_{\rho_{(k_1)}} < M
\end{equation}
for any $k_1, k_2 \in \{1,\ldots, K\}$. Hence, we obtain
\begin{equation}\label{rc_contrast}
r_2^c+\Delta r_{\pi_{(1)}}  > r_1^c +\Delta r_{\rho_{(K)}}
\end{equation}
and
\begin{equation}\label{con_r2c}
r_2^c + \Delta r_{\pi_{(K)}} < r^c_1 + \Delta r_{\rho_{(1)}}+M.
\end{equation}
Since $0\le r_1^c, r_2^c < M$, $ r_2^c - r_1^c \ge M/4$ and $|\Delta r_{l,k}|< M/8$, we have
\begin{equation}\label{minus_range}
-M/8 < r_1^c + \Delta r_{\rho_{(k_1)}} \le r_2^c -M/4 + \Delta r_{\rho_{(k_2)}} < M
\end{equation}
and
\begin{equation}\label{large_M}
0 <  r_1^c+M/4+ \Delta r_{\pi_{(k_1)}} \le r_2^c + \Delta r_{\pi_{(k_2)}} < 9M/8
\end{equation}
for any $k_1, k_2 \in \{1,\ldots, K\}$. By (\ref{condi_r1r2}) and (\ref{minus_range}), we obtain that if there exist some $k\in\{1,\ldots, K\}$ satisfying $r_1^c + \Delta r_{\rho_{(k)}}< 0$, then we have $r_2^c + \Delta r_{\pi_{(k_2)}} < M$ for any $k_2 \in \{1,\ldots, K\}$. By (\ref{condi_r1r2}) and (\ref{large_M}), we obtain that if there exist some $k\in\{1,\ldots, K\}$ satisfying $r_2^c + \Delta r_{\pi_{(k)}}> M$, then we have $r_1^c + \Delta r_{\rho_{(k_1)}} > 0$ for any $k_1 \in \{1,\ldots, K\}$. Hence, we have three cases below.

\noindent \textbf{Subcase 1:} $r_1^c + \Delta r_{\rho_{(k)}}< 0$ for some $k\in\{1,\ldots,K\}$.

Define $k'\in \{1,\ldots, K\}$ as
\begin{eqnarray*}
k' \triangleq
\begin{cases} K,& \text{if} ~ r_1^c +\Delta r_{\rho_{(K)}} < 0 \\
              \max\left\{k \colon r_1^c +\Delta r_{\rho_{(k)}} < 0, r_1^c +\Delta r_{\rho_{(k+1)}}\ge 0 \right\}, &\text{otherwise}.
\end{cases}
\end{eqnarray*}

\noindent 1) $k' = K$.

Combining (\ref{rc_contrast}) and (\ref{con_r2c}), we have
\begin{eqnarray}\label{case12}
0 < r_2^c + \Delta r_{\pi_{(1)}} \le \cdots
\le  r_2^c + \Delta r_{\pi_{(K)}} <
r^c_1 + \Delta r_{\rho_{(1)}}+M \le \cdots \le  r^c_1 +\Delta r_{\rho_{(K)}}+M < M.
\end{eqnarray}
From (\ref{r_lk_M}) and (\ref{case12}), we obtain that $\left\{r^c_1 + \Delta r_{\rho_{(1)}}+M, \ldots, r_1^c + \Delta r_{\rho_{(K)}}+M, r_2^c + \Delta r_{\pi_{(1)}}, \ldots, r_2^c + \Delta r_{\pi_{(K)}}\right\}$ are the $2K$ remainders $\left\{\tilde r_{1}^ c, \ldots, \tilde r_{2K}^ c\right\}$. Hence, (\ref{sequence_rc}) is equivalent to (\ref{case12}).
If we let $k_0=K$, then we have
\begin{eqnarray*}
D_{k_0} + D_{k_0+K} & = & \tilde r_{\varsigma_{(K+1)}}^c -\tilde r_{\varsigma_{(K)}}^c + \tilde r_{\varsigma_{(1)}}^c -\tilde r_{\varsigma_{(2K)}}^c +M \\
&=& r_1^c+\Delta r_{\rho_{(1)}}+M -r_2^c -\Delta r_{\pi_{(K)}} + r_2^c + \Delta r_{\pi_{(1)}} - r_1^c - \Delta r_{\rho_{(K)}}  \\
&=& \Delta r_{\rho_{(1)}}+M -\Delta r_{\pi_{(K)}} + \Delta r_{\pi_{(1)}} -\Delta r_{\rho_{(K)}} .
\end{eqnarray*}
Since $\left|\Delta r_{l,k}\right|< M/8$ for $l=1,2$; $k=1,\ldots, K$, we have
\begin{equation*}
D_{k_0} + D_{k_0+K}> M/2.
\end{equation*}
By the definitions of $\Omega_1$ and $\Omega_2$, we obtain from (\ref{case12}) that
\begin{eqnarray}\label{set1_1}
\begin{split}
& \Omega_1=\left\{r_1^c+\Delta r_{\rho_{(1)}}+M, \ldots, r_1^c + \Delta r_{\rho_{(K)}}+M \right\},\
& \Omega_2=\left\{r_2^c + \Delta r_{\pi_{(1)}}, \ldots, r_2^c +\Delta r_{\pi_{(K)}}\right\}.
\end{split}
\end{eqnarray}
Thus,
$$
\omega_K-\omega_1 \le 2\tau, \ \upsilon_K - \upsilon_1\le 2\tau.
$$

\noindent 2) $k' \in\{1,\ldots,K-1\}$.

Combining (\ref{rc_contrast}) and (\ref{con_r2c}), we have
\begin{align}\label{k_p_ineq}
 & 0 \le r_1^c + \Delta r_{\rho_{(k'+1)}}\le \cdots \le  r_1^c + \Delta r_{\rho_{(K)}} < r_2^c + \Delta r_{\pi_{(1)}} \le \cdots \nonumber \\
& \le  r_2^c + \Delta r_{\pi_{(K)}} < r^c_1 + \Delta r_{\rho_{(1)}}+M \le \cdots \le  r^c_1 +\Delta r_{\rho_{(k')}}+M < M.
\end{align}
From (\ref{r_lk_M}) and (\ref{k_p_ineq}), we obtain that $\left\{r^c_1 + \Delta r_{\rho_{(1)}}+M, \ldots, r^c_1 +\Delta r_{\rho_{(k')}}+M, r_1^c + \Delta r_{\rho_{(k'+1)}}, \ldots, r_1^c + \Delta r_{\rho_{(K)}}, \right. \\ r_2^c + \Delta r_{\pi_{(1)}}, \left. \ldots, r_2^c + \Delta r_{\pi_{(K)}}\right\}$ are the $2K$ remainders $\left\{\tilde r_{1}^ c, \ldots, \tilde r_{2K}^ c\right\}$. Hence, (\ref{sequence_rc}) is equivalent to (\ref{k_p_ineq}).
If we let $k_0= K -k'$, then we have
\begin{eqnarray*}
D_{k_0} + D_{k_0+K} & = & \tilde r_{\varsigma_{(K -k'+1)}}^c -\tilde r_{\varsigma_{(K -k')}}^c + \tilde r_{\varsigma_{(2K -k'+1)}}^c -\tilde r_{\varsigma_{(2K -k')}}^c \\
&=& r_2^c+\Delta r_{\pi_{(1)}} -r_1^c -\Delta r_{\rho_{(K)}} + r^c_1 + \Delta r_{\rho_{(1)}}+M - r_2^c - \Delta r_{\pi_{(K)}}  \\
&=& \Delta r_{\pi_{(1)}} -\Delta r_{\rho_{(K)}} + \Delta r_{\rho_{(1)}} -\Delta r_{\pi_{(K)}}+M \\
&>& M/2.
\end{eqnarray*}
By the definitions of $\Omega_1$ and $\Omega_2$, we obtain from (\ref{k_p_ineq}) that
\begin{eqnarray}\label{set_a2}
\begin{split}
& \Omega_1=\left\{r_2^c + \Delta r_{\pi_{(1)}}, \ldots, r_2^c +\Delta r_{\pi_{(K)}}\right\},\
& \Omega_2=\left\{r_1^c+\Delta r_{\rho_{(1)}}, \ldots, r_1^c + \Delta r_{\rho_{(K)}} \right\}.
\end{split}
\end{eqnarray}
Thus,
$$
\omega_K-\omega_1 \le 2\tau, \ \upsilon_K-\upsilon_1\le 2\tau.
$$

\noindent\textbf{Subcase 2:} $r_1^c + \Delta r_{\rho_{(k)}} \ge 0$ and $r_2^c + \Delta r_{\pi_{(k)}} < M$ for all $k\in\{1,\ldots,K\}$.

According to (\ref{rc_contrast}), we have
\begin{eqnarray}\label{ca2}
0 \le r_1^c + \Delta r_{\rho_{(1)}} \le \cdots
\le  r_1^c + \Delta r_{\rho_{(K)}} <
r_2^c + \Delta r_{\pi_{(1)}} \le \cdots \le  r_2^c +\Delta r_{\pi_{(K)}} < M.
\end{eqnarray}
From (\ref{r_lk_M}) and (\ref{ca2}), we obtain that $\left\{r^c_1 + \Delta r_{\rho_{(1)}}, \ldots, r_1^c + \Delta r_{\rho_{(K)}}, r_2^c + \Delta r_{\pi_{(1)}}, \ldots, r_2^c + \Delta r_{\pi_{(K)}}\right\}$ are the $2K$ remainders $\left\{\tilde r_{1}^ c, \ldots, \tilde r_{2K}^ c\right\}$. Hence, (\ref{sequence_rc}) is equivalent to (\ref{ca2}).
If we let $k_0= K $, then we have
\begin{eqnarray*}
D_{k_0} + D_{k_0+K} & = & \tilde r_{\varsigma_{(K+1)}}^c -\tilde r_{\varsigma_{(K)}}^c + \tilde r_{\varsigma_{(1)}}^c -\tilde r_{\varsigma_{(2K)}}^c +M \\
&=& r_2^c+\Delta r_{\pi_{(1)}} -r_1^c -\Delta r_{\rho_{(K)}} + r_1^c + \Delta r_{\rho_{(1)}} - r_2^c - \Delta r_{\pi_{(K)}}+M  \\
&=& \Delta r_{\pi_{(1)}} -\Delta r_{\rho_{(K)}} + \Delta r_{\rho_{(1)}} -\Delta r_{\pi_{(K)}}+M \\
&>& M/2.
\end{eqnarray*}
By the definitions of $\Omega_1$ and $\Omega_2$, we obtain from (\ref{ca2}) that
\begin{eqnarray}\label{set_a3}
\begin{split}
& \Omega_1=\left\{r_2^c + \Delta r_{\pi_{(1)}}, \ldots, r_2^c +\Delta r_{\pi_{(K)}}\right\},\
& \Omega_2=\left\{r_1^c+\Delta r_{\rho_{(1)}}, \ldots, r_1^c + \Delta r_{\rho_{(K)}} \right\}.
\end{split}
\end{eqnarray}
Thus,
$$
\omega_K -\omega_1 \le 2\tau, \ \upsilon_K -\upsilon_1\le 2\tau.
$$

\noindent\textbf{Subcase 3:} $r_2^c + \Delta r_{\pi_{(k)}} \ge M$ for some $k\in\{1,\ldots,K\}$.

Define $k''\in \{1,\ldots, K\}$ as
\begin{eqnarray*}
k'' \triangleq
\begin{cases} 1, & \text{if} \ r_2^c +\Delta r_{\pi_{(1)}} \ge M \\
              \min\left\{k \colon r_2^c +\Delta r_{\pi_{(k-1)}} < M, \ r_2^c +\Delta r_{\pi_{(k)}} \ge M\right\}, &\text{otherwise}.
\end{cases}
\end{eqnarray*}

\noindent 1) $k''=1$.

Combining (\ref{rc_contrast}) and (\ref{con_r2c}), we have
\begin{eqnarray}\label{case31}
0 \le r_2^c + \Delta r_{\pi_{(1)}}-M \le \cdots \le  r_2^c + \Delta r_{\pi_{(K)}}-M < r_1^c + \Delta r_{\rho_{(1)}} \le \cdots
\le  r_1^c + \Delta r_{\rho_{(K)}} < M.
\end{eqnarray}
From (\ref{r_lk_M}) and (\ref{case31}), we obtain that $\left\{r_1^c + \Delta r_{\rho_{(1)}}, \ldots, r_1^c + \Delta r_{\rho_{(K)}}, r_2^c + \Delta r_{\pi_{(1)}}-M, \ldots, r_2^c +\Delta r_{\pi_{(K)}}-M\right\}$ are the $2K$ remainders $\left\{\tilde r_{1}^ c, \ldots, \tilde r_{2K}^ c\right\}$. Hence, (\ref{sequence_rc}) is equivalent to (\ref{case31}).
If we let $k_0=K$, then we have
\begin{eqnarray*}
D_{k_0} + D_{k_0+K} & = & \tilde r_{\varsigma_{(K+1)}}^c -\tilde r_{\varsigma_{(K)}}^c + \tilde r_{\varsigma_{(1)}}^c -\tilde r_{\varsigma_{(2K)}}^c + M \\
&=& r_1^c+\Delta r_{\rho_{(1)}} -r_2^c -\Delta r_{\pi_{(K)}} +M + r_2^c + \Delta r_{\pi_{(1)}}-M - r_1^c - \Delta r_{\rho_{(K)}}+M \\
&=& \Delta r_{\rho_{(1)}} -\Delta r_{\pi_{(K)}} + \Delta r_{\pi_{(1)}} - \Delta r_{\rho_{(K)}}+M \\
&>& M/2.
\end{eqnarray*}
By the definitions of $\Omega_1$ and $\Omega_2$, we obtain from (\ref{case31}) that
\begin{eqnarray}\label{set_a4}
\begin{split}
& \Omega_1=\left\{r_1^c+\Delta r_{\rho_{(1)}}, \ldots, r_1^c + \Delta r_{\rho_{(K)}}\right\}, \
& \Omega_2=\left\{r_2^c + \Delta r_{\pi_{(1)}}-M, \ldots, r_2^c +\Delta r_{\pi_{(K)}}-M\right\}.
\end{split}
\end{eqnarray}
Thus,
$$
\omega_K- \omega_1 \le 2\tau, \ \upsilon_K - \upsilon_1\le 2\tau.
$$

\noindent 2) $k'' \in\{2,\ldots,K\}$.

Combining (\ref{rc_contrast}) and (\ref{con_r2c}), we have
\begin{align}\label{case32}
 & 0 \le r_2^c + \Delta r_{\pi_{(k'')}}-M \le \cdots \le  r_2^c + \Delta r_{\pi_{(K)}}-M < r_1^c + \Delta r_{\rho_{(1)}} \le \cdots \nonumber \\
& \le  r_1^c + \Delta r_{\rho_{(K)}} < r_2^c + \Delta r_{\pi_{(1)}} \le \cdots \le  r_2^c +\Delta r_{\pi_{(k''-1)}} < M.
\end{align}
From (\ref{r_lk_M}) and (\ref{case32}), we obtain that $\left\{r^c_1 + \Delta r_{\rho_{(1)}}, \ldots, r_1^c + \Delta r_{\rho_{(K)}}, r_2^c + \Delta r_{\pi_{(1)}}, \ldots, r_2^c + \Delta r_{\pi_{(k''-1)}},  \right. \\
\left. r_2^c + \Delta r_{\pi_{(k'')}}-M, \ldots, r_2^c + \Delta r_{\pi_{(K)}}-M \right\}$ are the $2K$ remainders $\left\{\tilde r_{1}^ c, \ldots, \tilde r_{2K}^ c\right\}$. Hence, (\ref{sequence_rc}) is equivalent to (\ref{case32}).
If we let $k_0=K - k''+1 $, then we have
\begin{eqnarray*}
D_{k_0} + D_{k_0+K} & = & \tilde r_{\varsigma_{(K-k''+2)}}^c -\tilde r_{\varsigma_{(K-k''+1)}}^c + \tilde r_{\varsigma_{(2K-k''+2)}}^c -\tilde r_{\varsigma_{(2K - k''+1)}}^c \\
&=& r_1^c+\Delta r_{\rho_{(1)}} -r_2^c -\Delta r_{\pi_{(K)}}+M + r_2^c + \Delta r_{\pi_{(1)}}- r_1^c - \Delta r_{\rho_{(K)}} \\
&=& \Delta r_{\rho_{(1)}} -\Delta r_{\pi_{(K)}} + \Delta r_{\pi_{(1)}} -\Delta r_{\rho_{(K)}}+M \\
&>& M/2.
\end{eqnarray*}
By the definitions of $\Omega_1$ and $\Omega_2$, we obtain from (\ref{case32}) that
\begin{eqnarray}\label{set_a5}
\begin{split}
& \Omega_1=\left\{r_1^c+\Delta r_{\rho_{(1)}}, \ldots, r_1^c + \Delta r_{\rho_{(K)}}\right\}, \
& \Omega_2=\left\{r_2^c + \Delta r_{\pi_{(1)}}-M, \ldots, r_2^c +\Delta r_{\pi_{(K)}}-M\right\}.
\end{split}
\end{eqnarray}
Thus,
$$
\omega_K-\omega_1 \le 2\tau, \ \upsilon_K -\upsilon_1\le 2\tau.
$$

Next, we prove that $k_0$ is the only subscript satisfying (\ref{d_k}). In fact, for any $k^*\in \{1,\ldots, K\} \setminus \{k_0\}$, we obtain from (\ref{sum_dk}) that
$$
D_{k^*} + D_{k^*+K} \le  \sum_{k\ne k_0}\big(D_{k}+D_{k+K}\big)
= M - \big(D_{k_0} + D_{k_0+K}\big) 
< M/2.
$$
This completes the proof.
\end{proof}

\subsection{Proof of Corollary \ref{Omega_M4}}
\begin{proof}
As we obtained in Case 2 in the above proof of Lemma \ref{d_kM8}, the two clusters $\Omega_1$ and $\Omega_2$ are given as
\begin{eqnarray*}\label{corollary_case}
\begin{cases} r_1^c + \Delta r_{\rho_{(K)}}< 0, & \text{see} \ (\ref{set1_1}) \\
              r_1^c + \Delta r_{\rho_{(k)}}< 0 \ \text{for some} \ k \in\{1,\ldots,K-1\}, & \text{see} \ (\ref{set_a2}) \\
              r_1^c + \Delta r_{\rho_{(k)}} \ge 0, \ r_2^c + \Delta r_{\pi_{(k)}} < M \ \text{for} \ 1\le k \le K, & \text{see} \ (\ref{set_a3}) \\
              r_2^c + \Delta r_{\pi_{(1)}} \ge M, & \text{see} \ (\ref{set_a4}) \\
              r_2^c + \Delta r_{\pi_{(k)}} \ge M \ \text{for some} \ k \in \{2,\ldots,K\}, & \text{see} \ (\ref{set_a5}).
\end{cases}
\end{eqnarray*}
Since the proofs of the five cases are similar, we only consider the case $r_1^c + \Delta r_{\rho_{(K)}}< 0$ in the following.

Recall that $\{\Delta r_{\rho_{(1)}}, \ldots,\Delta r_{\rho_{(K)}}\}$ are sorted in the increasing order from the remainder errors $\{\Delta r_{1,1}, \ldots, \Delta r_{1,K}\}$. Hence, for each $\Delta r_{1,k}\in \{\Delta r_{1,1}, \ldots, \Delta r_{1,K}\}$, there exists $k_1 \in \{1,\ldots, K\}$ such that
\begin{equation}\label{delta_rho_rk}
\Delta r_{1,k} = \Delta r_{\rho_{(k_1)}}.
\end{equation}
According to (\ref{Omega_1_2}) and (\ref{set1_1}), we have
$$
\Omega_1= \{\omega_1,\ldots, \omega_K\} =\{r_1^c+\Delta r_{\rho_{(1)}}+M, \ldots, r_1^c+\Delta r_{\rho_{(K)}}+M \}.
$$
Since $\omega_1 \le \cdots \le \omega_K$ and $r_1^c+\Delta r_{\rho_{(1)}}+M \le \cdots \le r_1^c+\Delta r_{\rho_{(K)}}+M$, we have
\begin{equation}\label{omega_k}
r_1^c+\Delta r_{\rho_{(k)}}+M = \omega_k, \ k=1,\ldots, K.
\end{equation}
By (\ref{delta_rho_rk}) and (\ref{omega_k}), we obtain
$$
r_1^c+\Delta r_{1,k}+M = r_1^c+\Delta r_{\rho_{(k_1)}}+M = \omega_{k_1}.
$$
By the definition of circular distance in (\ref{eq:define_remainder_distance}), we obtain
\begin{equation}\label{coro_con1}
d_M(r_1^c+\Delta r_{1,k}, \omega_{k_1})=0.
\end{equation}
By (\ref{r_lk_q}), we have
$$
d_M(\tilde r_{1,k}^c,\omega_{k_1})= d_M(\tilde r_{1,k},\omega_{k_1})=d_M(r_1^c+\Delta r_{1,k}, \omega_{k_1}) =0.
$$
Similarly, for each $\Delta r_{2,k}\in \{\Delta r_{2,1}, \ldots, \Delta r_{2,K}\}$, we can prove that there exists $\upsilon_{k_2} \in \Omega_2$, $k_2\in\{1,\ldots,K\}$ satisfying
$$
d_M(\tilde r_{2,k}^c,\upsilon_{k_2})=0.
$$
Therefore, (\ref{condi1}) holds.
\end{proof}

\subsection{Proof of Theorem \ref{Th_N_estimate}}
\begin{proof}
From  (\ref{calcu_q_k}) and (\ref{r_lk_q}), we obtain
\begin{eqnarray}\label{hat_q_tk}
\hat q_{l,k} &=& \left[ \frac{Mq_{l,k}+r^c_l +\Delta r_{l,k}-\overline \omega_t}{M}\right] \nonumber \\
&=& q_{l,k} + \left[ \frac{r^c_l +\Delta r_{l,k}-\overline \omega_t}{M}\right],
\end{eqnarray}
where $l=1,2; k=1,\ldots, K$, and $t$ is defined in (\ref{def_t}). For convenience, we denote
$\overline {\Delta r}_l \triangleq \frac {1}{K}\sum_{k=1}^{K}\Delta r_{l,k}$, $l=1,2$. Clearly, $|\overline {\Delta r}_l|\le \tau$.

\noindent \textbf{Case 1:}  $0 \le \left| d_M(r_1^c, r_2^c) \right| < M/4$.

As we obtained in Case 1 of Lemma \ref{d_kM8}, the two clusters $\Omega_1$ and $\Omega_2$ are given as
\begin{eqnarray}
\begin{cases} \ell_{2K}-\ell_1=0, & \text{see} \ (\ref{set2_1}) \\
              \ell_{j+1}-\ell_j = 1 \ \text{for some}\ j< K, & \text{see} \ (\ref{set2_21}) \\
              \ell_{j+1}-\ell_j = 1 \ \text{for some}\ j\ge K, & \text{see} \ (\ref{set2_22}),
\end{cases}
\end{eqnarray}
where $\ell_i$ are defined in (\ref{b_presentation}). Since the proofs of the three cases are similar, we only prove the case $\ell_{2K}-\ell_1=0$.

By (\ref{set_r1c}) and (\ref{b_presentation}), we obtain that $\ell_1= \cdots =\ell_{2K}=-1$, $0$, or $1$. Since the proofs of the three cases are similar, we only consider the case $\ell_1= \cdots =\ell_{2K}=0$, and $c_i$ are described for the case $3M/4 < r^c_2-r^c_1 < M$. According to (\ref{set2_1}), we have
\begin{eqnarray}\label{sets_Omega_12}
\begin{split}
& \Omega_1= \{\omega_1,\ldots,\omega_K\}=\left\{c_{K+1}, \ldots, c_{2K}\right\}, \
& \Omega_2= \{\upsilon_1,\ldots,\upsilon_K\} =\left\{c_1, \ldots, c_{K}\right\}.
\end{split}
\end{eqnarray}
Since $\omega_i \le \omega_j$, $\upsilon_i \le \upsilon_j$, and $c_i \le c_j$ for $1\le i < j\le K$, we obtain
\begin{equation}\label{c_omega}
\omega_K - \upsilon_1 = c_{2K}-c_1.
\end{equation}
Recall that $c_{2K}-c_1< M/2$ as previously shown in (\ref{b_dis}). Then, we have $\omega_K - \upsilon_1 < M/2$. According to (\ref{omega_prime}), we have
\begin{equation}\label{omega_prime_equal}
\omega'_k=\omega_k, \ k=1,\ldots,K.
\end{equation}
By the definitions of $\overline \omega_1$ and $\overline \omega_2$ in (\ref{overline_omega}), we obtain
\begin{equation}\label{omega_temp_12}
\overline \omega_1 = \frac{c_{K+1}+\cdots+c_{2K}}{K}, \ \overline \omega_2 = \frac{c_{1}+\cdots+c_{K}}{K}.
\end{equation}
According to the definitions of $c_i$ in (\ref{set_r1c}), we have
\begin{equation}\label{c_min_max}
c_{1}=\min\{r_1^c+\Delta r_{\rho_{(1)}},r_2^c+\Delta r_{\pi_{(1)}}-M\}, \ c_{2K}=\max\{r_1^c+\Delta r_{\rho_{(K)}},r_2^c+\Delta r_{\pi_{(K)}}-M\}.
\end{equation}
Then, we have four cases below.

\noindent \textcircled {1} $c_1= r_2^c+\Delta r_{\pi_{(1)}}-M$ and $c_{2K}=r_1^c+\Delta r_{\rho_{(K)}}$.

Since $c_{2K}=r_1^c+\Delta r_{\rho_{(K)}}$, we obtain from the definitions of $c_i$ in (\ref{set_r1c}) that $c_{K+1} \ge r_1^c+\Delta r_{\rho_{(1)}}$. From (\ref{omega_temp_12}), we have
$$
r_1^c+\Delta r_{\rho_{(1)}} \le c_{K+1}\le \overline \omega_1 \le c_{2K}= r_1^c+\Delta r_{\rho_{(K)}}.
$$
Note that $|\Delta r_{l,k}|\le \tau$ for $l=1,2;k=1,\ldots,K$. Hence, there exists $|\epsilon_1| \le \tau$ satisfying
\begin{equation}\label{omega_11}
\overline \omega_1 = r_1^c+\epsilon_1.
\end{equation}

Since $c_1= r_2^c+\Delta r_{\pi_{(1)}}-M$, we obtain from the definitions of $c_i$ in (\ref{set_r1c}) that $c_{K} \le r_2^c+\Delta r_{\pi_{(K)}}-M$. From (\ref{omega_temp_12}), we have
$$
r_2^c+\Delta r_{\pi_{(1)}}-M =c_1 \le \overline \omega_2 \le c_{K}\le r_2^c+\Delta r_{\pi_{(K)}}-M.
$$
Hence, there exists $|\epsilon_2|\le \tau$ satisfying
\begin{equation}\label{omega_21}
\overline \omega_2 = r_2^c- M +\epsilon_2.
\end{equation}

\noindent \textcircled {2} $c_1= r_1^c+\Delta r_{\rho_{(1)}}$ and $c_{2K}=r_2^c+\Delta r_{\pi_{(K)}}-M$.

Similarly, we obtain from (\ref{omega_temp_12}) that
\begin{eqnarray*}
\begin{split}
& r_2^c+\Delta r_{\pi_{(1)}}-M \le c_{K+1} \le \overline \omega_1 \le c_{2K}=r_2^c+\Delta r_{\pi_{(K)}}-M, \\
& r_1^c+\Delta r_{\rho_{(1)}} =c_1 \le \overline \omega_2 \le c_{K}\le r_1^c+\Delta r_{\rho_{(K)}}.
\end{split}
\end{eqnarray*}
Hence, there exist $|\epsilon_3 |\le \tau$ and $|\epsilon_4|\le \tau$ satisfying
\begin{equation}\label{omega_12}
\overline \omega_1 = r_2^c -M +\epsilon_3, \  \overline \omega_2 = r_1^c +\epsilon_4.
\end{equation}

\noindent \textcircled {3} $c_1= r_1^c+\Delta r_{\rho_{(1)}}$ and $c_{2K}=r_1^c+\Delta r_{\rho_{(K)}}$.

Similar to \textcircled {1}, we can prove that $\overline \omega_1$ is the same as (\ref{omega_11}).

Note that $3M/4 < r_2^c -r_1^c < M$. Hence,
\begin{equation}\label{dif_r2_1}
r_2^c-M < r_1^c.
\end{equation}
Since $c_1 = r_1^c+\Delta r_{\rho_{(1)}}$, we have $c_K < r_2^c-M+\Delta r_{\pi_{(K)}}$. Otherwise, $c_K \ge r_2^c-M+\Delta r_{\pi_{(K)}}$. Then, we have $c_1 \ge r_2^c-M+\Delta r_{\pi_{(1)}}$, which is a contradiction. By (\ref{omega_temp_12}), we obtain
\begin{equation}\label{addi1}
r_1^c+\Delta r_{\rho_{(1)}}=c_1 \le \overline \omega_2\le c_K < r_2^c-M+\Delta r_{\pi_{(K)}}.
\end{equation}
Since $\left|\Delta r_{l,k}\right| \le \tau$ for $l=1,2; k=1,\ldots,K$, we obtain from (\ref{dif_r2_1}) that
\begin{equation}\label{addi2}
r_2^c-M-\tau < r_1^c-\tau \le r_1^c+\Delta r_{\rho_{(1)}}.
\end{equation}
Combining (\ref{addi1}) and (\ref{addi2}), we have
$$
r_2^c-M-\tau <\overline \omega_2 < r_2^c-M+\Delta r_{\pi_{(K)}}.
$$
Hence, there exists $|\epsilon_5|\le \tau$ satisfying
\begin{equation}\label{omega_new1}
\overline \omega_2 = r_2^c - M +\epsilon_5.
\end{equation}

\noindent \textcircled {4} $c_1= r_2^c+\Delta r_{\pi_{(1)}}-M$ and $c_{2K}=r_2^c+\Delta r_{\pi_{(K)}}-M$.

Similar to \textcircled {1}, we can prove that $\overline \omega_2$ is the same as (\ref{omega_21}).

Since $c_{2K}=r_2^c+\Delta r_{\pi_{(K)}}-M$, we have $c_{K+1} > r_1^c +\Delta r_{\rho_{(1)}}$. Otherwise, $c_{K+1} \le r_1^c +\Delta r_{\rho_{(1)}}$. Then, we have $c_{2K} \le r_1^c +\Delta r_{\rho_{(K)}}$, which is a contradiction. By (\ref{omega_temp_12}), we obtain
\begin{equation}\label{addi3}
r_1^c+\Delta r_{\rho_{(1)}} < c_{K+1} \le \overline \omega_1 \le c_{2K} = r_2^c+\Delta r_{\pi_{(K)}}-M.
\end{equation}
Since $\left|\Delta r_{l,k}\right| \le \tau$ for $l=1,2; k=1,\ldots,K$, we obtain from (\ref{dif_r2_1}) that
\begin{equation}\label{addi4}
r_2^c+\Delta r_{\pi_{(K)}}-M \le r_2^c +\tau -M < r_1^c +\tau.
\end{equation}
Combining (\ref{addi3}) and (\ref{addi4}), we have
$$
r_1^c+\Delta r_{\rho_{(1)}} < \overline \omega_1 < r_1^c +\tau.
$$
Hence, there exists $|\epsilon_6|\le \tau$ satisfying
\begin{equation}\label{omega_new2}
\overline \omega_1 = r_1^c +\epsilon_6.
\end{equation}
Now, we check $\{\hat q_{1,k},\hat q_{2,k}\}$, $k=1,\ldots,K$.

According to (\ref{calcu_q_k}), either $\overline \omega_1$ or $\overline \omega_2$ is subtracted from $\tilde r_{l,k}$. Note that either $d_M(\tilde r^c_{l,k},\omega_{k_1})=0$ or $d_M(\tilde r^c_{l,k},\upsilon_{k_2})=0$ holds for some $k_1,k_2\in \{1,\ldots,K\}$. Hence, we obtain from (\ref{hat_q_tk}) that $\hat q_{l,k}$ is either
\begin{eqnarray}\label{q_est_AB}
\hat q_{l,k} = q_{l,k} + \left[ \frac{r^c_l +\Delta r_{l,k}-\overline \omega_1}{M}\right]
\text{or}\
\hat q_{l,k} = q_{l,k} + \left[ \frac{r^c_l +\Delta r_{l,k}-\overline \omega_2}{M}\right].
\end{eqnarray}

When $l=1$, $\overline \omega_1 = r_1^c+\epsilon_1$ and $\overline \omega_2 = r_2^c -M +\epsilon_2$, we have
\begin{eqnarray*}
\begin{split}
& r^c_1 +\Delta r_{1,k}-\overline \omega_1 =\Delta r_{1,k}-\epsilon_1, \
& r^c_1 +\Delta r_{1,k}-\overline \omega_2 = r^c_1 +\Delta r_{1,k}-r^c_2+M-\epsilon_2.
\end{split}
\end{eqnarray*}
Since $3M/4 < r^c_2-r^c_1 < M$, $\left|\Delta r_{1,k}\right|\le \tau$, and $\tau < M/8 $, we obtain
$$
-M/4 < \Delta r_{1,k}-\epsilon_1< M/4, \  -M/4 < r^c_1 +\Delta r_{1,k}-r^c_2+M-\epsilon_2 < M/2.
$$
Hence,
$$
\left[ \frac{r^c_1 +\Delta r_{1,k}-\overline \omega_1}{M}\right]=0, \ \left[ \frac{r^c_1 +\Delta r_{1,k}-\overline \omega_2}{M}\right]=0.
$$
It follows from (\ref{q_est_AB}) that
\begin{equation}\label{est_q_1k}
\hat q_{1,k}= q_{1,k}, \ k=1,\ldots,K.
\end{equation}
Similarly, for the cases $\overline \omega_1 = r_2^c-M+\epsilon_3$ and $\overline \omega_2 = r_1^c+\epsilon_4$; $\overline \omega_1 = r_1^c +\epsilon_1$ and $\overline \omega_2 = r_2^c -M +\epsilon_5$; $\overline \omega_1 = r_1^c +\epsilon_6$ and $\overline \omega_2 = r_2^c-M+\epsilon_2$, we can also obtain (\ref{est_q_1k}).

When $l=2$, $\overline \omega_1 = r_1^c+\epsilon_1$ and $\overline \omega_2 = r_2^c -M +\epsilon_2$, we have
\begin{eqnarray*}
\begin{split}
& r^c_2 +\Delta r_{2,k}- \overline \omega_1 = r^c_2 +\Delta r_{2,k}-r^c_1 -\epsilon_1, \
& r^c_2 +\Delta r_{2,k}-\overline \omega_2 = \Delta r_{2,k}+M-\epsilon_2.
\end{split}
\end{eqnarray*}
Note that
$$
 M/2 < r^c_2 +\Delta r_{2,k}-r^c_1 -\epsilon_1 < 5M/4, \  3M/4 < \Delta r_{2,k}+M-\epsilon_2 < 5M/4.
$$
Then,
$$
\left[ \frac{r^c_2 +\Delta r_{2,k}-\overline \omega_1}{M}\right]=1, \ \left[ \frac{r^c_2 +\Delta r_{2,k}-\overline \omega_2}{M}\right]=1.
$$
By (\ref{q_est_AB}), we have
\begin{equation}\label{est_q_2k}
\hat q_{2,k}= q_{2,k}+1, \ k=1,\ldots,K.
\end{equation}
Similarly, for the cases $\overline \omega_1 = r_2^c-M+\epsilon_3$ and $\overline \omega_2 = r_1^c+\epsilon_4$; $\overline \omega_1 = r_1^c +\epsilon_1$ and $\overline \omega_2 = r_2^c -M +\epsilon_5$; $\overline \omega_1 = r_1^c +\epsilon_6$ and $\overline \omega_2 = r_2^c-M+\epsilon_2$, we can also obtain (\ref{est_q_2k}).

Therefore, $\{\hat q_{1,k},\hat q_{2,k}\}= \{ q_{1,k}, q_{2,k}+1\}$, $k=1,\ldots,K$. By the generalized CRT for two integers obtained in \cite{wangwei2014}, we have
$$
\{\hat Q_1, \hat Q_2 \}= \{Q_1, Q_2+1 \}.
$$
Next, we check $\{\hat N_1, \hat N_2\}$ for the cases $\hat Q_1 = \hat Q_2$ and $\hat Q_1 \ne \hat Q_2$.

\noindent 1) $\hat Q_1 = \hat Q_2= Q_1 = Q_2+1$.

We have four cases below.

\noindent \textcircled {1} $c_1= r_2^c+\Delta r_{\pi_{(1)}}-M$ and $c_{2K}=r_1^c+\Delta r_{\rho_{(K)}}$.

By (\ref{estimate_NAB_eq}), (\ref{omega_11}), and (\ref{omega_21}), we have
$$
\{\hat N_1, \hat N_2\} = \{ N_1 +\epsilon_1,  N_2 +\epsilon_2\}.
$$

\noindent \textcircled {2} $c_1= r_1^c+\Delta r_{\rho_{(1)}}$ and $c_{2K}=r_2^c+\Delta r_{\pi_{(K)}}-M$.

By (\ref{estimate_NAB_eq}) and (\ref{omega_12}), we have
$$
\{\hat N_1, \hat N_2\} =\{ N_1 +\epsilon_4, N_2 +\epsilon_3 \}.
$$

\noindent\textcircled{3} $c_{1}= r_1^c+\Delta r_{\rho_{(1)}}$ and $c_{2K}= r_1^c+\Delta r_{\rho_{(K)}}$.

By (\ref{estimate_NAB_eq}), (\ref{omega_11}), and (\ref{omega_new1}), we have
$$
\{\hat N_1, \hat N_2\} = \{ N_1 +\epsilon_1,  N_2 +\epsilon_5\}.
$$

\noindent\textcircled{4} $c_{1}= r_2^c+\Delta r_{\pi_{(1)}}-M$ and $c_{2K}=r_2^c+\Delta r_{\pi_{(K)}}-M$.

By (\ref{estimate_NAB_eq}), (\ref{omega_21}), and (\ref{omega_new2}), we have
$$
\{\hat N_1, \hat N_2\} = \{ N_1 +\epsilon_6,  N_2 +\epsilon_2\}.
$$

Therefore, (\ref{N_robust}) holds.

\noindent 2) $\hat Q_1 \ne \hat Q_2$.

For simplicity, we suppose that $\hat q_{1,k} \ne \hat q_{2,k}$ for all $k=1,\ldots, K$. By (\ref{def_Omega_prime}), (\ref{sets_Omega_12}), and (\ref{omega_prime_equal}), we obtain
\begin{eqnarray*}
\Omega^\prime & =& \{\omega_1, \ldots, \omega_K, \upsilon_1, \ldots, \upsilon_K\}=\{c_1, \ldots, c_{2K}\} \\
& = & \left\{r_1^c+\Delta r_{1,1}, \ldots, r_1^c+\Delta r_{1,K}, r_2^c+\Delta r_{2,1}-M, \ldots, r_2^c+\Delta r_{2,K}-M\right\}.
\end{eqnarray*}
According to (\ref{r_lk_q}), we have
\begin{eqnarray*}
\begin{split}
 d_M(\tilde r_{1,k},r_1^c+\Delta r_{1,k}) = 0,\ d_M(\tilde r_{2,k},r_2^c+\Delta r_{2,k}-M) =0.
\end{split}
\end{eqnarray*}
It follows from (\ref{r_sigma_eta_2}) that
$$
\hat r_{1,k}^c=r_1^c+\Delta r_{1,k}, \ \hat r_{2,k}^c=r_2^c+\Delta r_{2,k}-M.
$$
From (\ref{rc_AB}), we obtain
$$
\hat r^c_1 =  r^c_1+\overline {\Delta r}_1, \ \hat r^c_2 =  r^c_2-M +\overline {\Delta r}_2.
$$
By (\ref{estimate_NAB}), we have
$$
\{\hat N_1, \hat N_2\} = \{N_1 +\overline {\Delta r}_1,N_2 +\overline {\Delta r}_2\}.
$$
Therefore, (\ref{N_robust}) holds.

\noindent \textbf{Case 2:}  $M/4 \le \left| d_M(r_1^c, r_2^c) \right| \le M/2$.

In this case, all the possible cases of the two clusters, $\Omega_1$ and $\Omega_2$, are described in (\ref{corollary_case}). Since the proofs of the five cases are similar, we only consider the case $r_1^c + \Delta r_{\rho_{(K)}}< 0$ in the following.

It is noted that the two clusters are given by (\ref{set1_1}), i.e.,
\begin{eqnarray*}
\begin{split}
&\Omega_1= \{\omega_1,\ldots,\omega_K\}=\left\{r_1^c+\Delta r_{\rho_{(1)}}+M , \ldots, r_1^c+\Delta r_{\rho_{(K)}}+M \right\}, \\
& \Omega_2= \{\upsilon_1,\ldots,\upsilon_K\} =\left\{r_2^c+\Delta r_{\pi_{(1)}}, \ldots, r_2^c+\Delta r_{\pi_{(K)}}\right\}.
\end{split}
\end{eqnarray*}
Since $\omega_K-\upsilon_1=r_1^c+\Delta r_{\rho_{(K)}}+M - (r_2^c+\Delta r_{\pi_{(1)}})$ and $M/4 \le r_2^c -r_1^c \le 3M/4$, we have
\begin{equation}
0< \omega_K - \upsilon_1 < M.
\end{equation}
Then, we have two cases: $0< \omega_K-\upsilon_1 \le M/2$ and $M/2 < \omega_K -\upsilon_1 < M$. Since the proofs of the two cases are similar, we only consider the case $M/2 < \omega_K -\upsilon_1 < M$. By the definitions of $\omega_k^\prime$ in (\ref{omega_prime}), we have
\begin{equation}\label{omega_k_ineq}
\omega_k^\prime=\omega_k-M =r_1^c + \Delta r_{1,k}, \ k=1,\ldots, K.
\end{equation}
According to the definitions of $\overline \omega_1$ and $\overline \omega_2$ in (\ref{overline_omega}), we obtain
\begin{equation}\label{omega_tem2}
\overline \omega_1 = r_1^c +\overline {\Delta r}_1, \ \overline \omega_2 = r_2^c+\overline {\Delta r}_2,
\end{equation}
Now, we check $\{\hat q_{1,k},\hat q_{2,k}\}$, $k=1,\ldots,K$.

According to (\ref{calcu_q_k}), either $\overline \omega_1$ or $\overline \omega_2$ is subtracted from $\tilde r_{l,k}$. By Corollary \ref{Omega_M4}, we know that (\ref{condi1}) holds. Hence, we obtain from (\ref{hat_q_tk}) that
\begin{eqnarray*}
\begin{split}
& \hat q_{1,k} = \left[ \frac{\tilde r_{1,k}-\overline \omega_1}{M}\right]= q_{1,k}+ \left[ \frac{\Delta r_{1,k}-\overline {\Delta r}_1}{M}\right] = q_{1,k}, \\
& \hat q_{2,k} = \left[ \frac{\tilde r_{2,k}-\overline \omega_2}{M}\right]= q_{2,k}+ \left[ \frac{\Delta r_{2,k}-\overline {\Delta r}_2}{M}\right] = q_{2,k}.
\end{split}
\end{eqnarray*}
Therefore, $\{\hat q_{1,k},\hat q_{2,k}\}= \{ q_{1,k}, q_{2,k}\}$, $ k=1,\ldots,K$. By the generalized CRT for two integers obtained in \cite{wangwei2014}, we have
$$
\{\hat Q_1, \hat Q_2\}=\{Q_1, Q_2\}.
$$
Next, we check $\{\hat N_1, \hat N_2\}$ for the cases $\hat Q_1 = \hat Q_2$ and $\hat Q_1 \ne \hat Q_2$.

\noindent 1) $\hat Q_1 = \hat Q_2$.

By (\ref{estimate_NAB_eq}) and (\ref{omega_tem2}), we have
$
\{\hat N_1, \hat N_2\} = \{N_1 + \overline {\Delta r}_1, N_2 +\overline {\Delta r}_2\}.
$
Hence, (\ref{N_robust}) holds.

\noindent 2) $\hat Q_1 \ne \hat Q_2$.

For simplicity, we suppose that $\hat q_{1,k} \ne \hat q_{2,k}$ for all $k=1,\ldots, K$. By (\ref{def_Omega_prime}) and (\ref{omega_k_ineq}), we obtain
\begin{eqnarray*}
\Omega^\prime & = & \{\omega_1^\prime, \ldots, \omega_K^\prime, \upsilon_1, \ldots, \upsilon_K\} \\
& = & \left\{r_1^c+\Delta r_{1,1}, \ldots, r_1^c+\Delta r_{1,K}, r_2^c+\Delta r_{2,1}, \ldots, r_2^c+\Delta r_{2,K} \right\}.
\end{eqnarray*}
According to (\ref{r_lk_q}), we have
\begin{eqnarray*}
d_M(\tilde r_{1,k},r_1^c+\Delta r_{1,k}) = 0, \
d_M(\tilde r_{2,k},r_2^c+\Delta r_{2,k})=0.
\end{eqnarray*}
It follows from (\ref{r_sigma_eta_2}) that
$$
\hat r_{1,k}^c=r_1^c+\Delta r_{1,k}, \ \hat r_{2,k}^c=r_2^c+\Delta r_{2,k}.
$$
According to (\ref{rc_AB}), we obtain
$$
\hat r^c_1 =  r^c_1 +\overline {\Delta r}_1, \ \hat r^c_2 =  r^c_2 +\overline {\Delta r}_2.
$$
By (\ref{estimate_NAB}), we have
$$
\{\hat N_1, \hat N_2\} = \{N_1 +\overline {\Delta r}_1, N_2 +\overline {\Delta r}_2\}.
$$
Hence, (\ref{N_robust}) holds. This completes the proof of the theorem.
\end{proof}

\small
\bibliographystyle{ieee}


\end{document}